\documentclass[11pt,a4paper,reqno]{amsart}%
\usepackage{amsthm,amsmath,amsfonts,amssymb,amsxtra,appendix,bookmark,euscript,delarray,dsfont}
\usepackage[latin1]{inputenc}

\theoremstyle{plain}
\newtheorem{theorem}{Theorem}
\newtheorem{lemma}[theorem]{Lemma}
\newtheorem{corollary}[theorem]{Corollary}

\theoremstyle{definition}

\theoremstyle{remark}
\newtheorem{remark}[theorem]{Remark}

\DeclareMathOperator{\Tr}{Tr}
\DeclareMathOperator{\tr}{Tr}

\def\leqslant{\le}
\def\bq{\begin{eqnarray}}
\def\eq{\end{eqnarray}}
\def\bqq{\begin{eqnarray*}}
\def\eqq{\end{eqnarray*}}
\def\nn{\nonumber}

\def\eps{\varepsilon}
\def\wto{\rightharpoonup}
\newcommand{\norm}[1]{\left\lVert #1 \right\rVert}

\renewcommand{\epsilon}{\varepsilon}
\renewcommand{\phi}{\varphi}
\newcommand\1{{\ensuremath {\mathds 1} }}

\def\cF {\mathcal{F}}

\def\cH{\mathcal{H}}
\def\bH{\mathbb{H}}

\def\R {\mathbb{R}}
\def\C {\mathbb{C}}
\def\N {\mathcal{N}}

\def\cL {\mathcal{L}}

\def\cN{\mathcal{N}}

\def\F {\mathcal{F}}

\def\H{\gH}

\def\R {\mathbb{R}}
\def\C {\mathbb{C}}
\def\N {\mathcal{N}}

\def\d{{\rm d}}
\newcommand{\gH}{\mathfrak{H}}
\newcommand\ii{{\ensuremath {\infty}}}
\newcommand\pscal[1]{{\ensuremath{\left\langle #1 \right\rangle}}}
\newcommand{\dGamma}{{\ensuremath{\rm d}\Gamma}}

\title[Fluctuations around Hartree states]{Fluctuations around Hartree states in the mean-field regime}

\author[M. Lewin]{Mathieu Lewin}
\address{CNRS \& Universit\'e de Cergy-Pontoise, Mathematics Department (UMR 8088), F-95000 Cergy-Pontoise, France} 
\email{mathieu.lewin@math.cnrs.fr}

\author[P.~T. Nam]{Phan Th\`anh Nam}
\address{CNRS \& Universit\'e de Cergy-Pontoise, Mathematics Department (UMR 8088), F-95000 Cergy-Pontoise, France} 
\email{phan-thanh.nam@u-cergy.fr}

\author[B. Schlein]{Benjamin Schlein}
\address{Institute for Applied Mathematics, University of Bonn, Bonn, 53115 Germany} 
\email{benjamin.schlein@hcm.uni-bonn.de}

\begin{document}
\date{August 26, 2014}

\begin{abstract} We consider the dynamics of a large system of $N$ interacting bosons in the mean-field regime where the interaction is of order $1/N$. We prove that the fluctuations around the nonlinear Hartree state are generated by an effective quadratic Hamiltonian in Fock space, which is derived from Bogoliubov's approximation. We use a direct method in the $N$-particle space, which is different from the one based on coherent states in Fock space.
\end{abstract}

\maketitle

\setcounter{tocdepth}{2}
\tableofcontents

\section{Introduction}

In this paper, we consider the dynamics of a system of $N$ bosons living in $\R^d$ in the so-called \emph{mean-field regime} and we are interested in the time-dependent fluctuations around the Hartree state solution. 

The $N$-body system is described by a time-dependent wave function $\Psi_N(t,x_1,...,x_N)$ in the symmetric space $\H^N=\bigotimes_{\text{sym}}^N L^2(\R^d)$ and its evolution is governed by Schr\"odinger's equation
\bq \label{eq:Schrodinger-dynamics}
\left\{ \begin{gathered}
 i\, \dot{\Psi}_{N}(t) = H_N \Psi_{N}(t),  \hfill \\
  \Psi_{N}(0)=\Psi_{N,0},\hfill \\ 
\end{gathered}  \right.\eq
or, equivalently,
$$\Psi_N(t)=e^{-it H_N}\Psi_N(0).$$
The many-particle Hamiltonian $H_N$ is given by
\begin{equation} 
H_N= \sum\limits_{j = 1}^N (-\Delta)_{x_j} + \frac{1}{N-1} \sum\limits_{1 \leqslant j < k \leqslant N} {w(x_j-x_k)}, 
\label{eq:def_H_N}
\end{equation}
where the Laplacian $-\Delta_{x_j}$ describes the kinetic energy of the $j$-th particle and $w:\R^d \to \R$ is a measurable, even function describing the interactions between the particles. 

The fact that we are considering the mean-field regime is apparent in the factor $1/(N-1)$ which makes the one-body term and the two-body interaction term of order $N$ in the Hamiltonian $H_N$. We could as well take $1/N$ instead of $1/(N-1)$ but the latter choice simplifies some expressions.

We shall always assume that the interaction potential $w$ satisfies the operator inequality 
\bq \label{eq:asp-w} w^2 \le C(1-\Delta)
\eq
on $L^2(\R^d)$, for some constant $C>0$. The condition (\ref{eq:asp-w}) in particular ensures that $w$ is relatively form bounded by the Laplacian with bound as small as we want:
$$ |w| \le \sqrt{C(1-\Delta)} \le \eps (1-\Delta)+ C \eps^{-1}$$
for every $\eps>0$. The reader may keep in mind the example of $w$ being the Coulomb potential $|x|^{-1}$ in $\R^3$. We may also consider bosons living in a bounded domain of $\R^d$ (with appropriate boundary conditions), or replace the Laplacian $-\Delta$ by the Schr\"odinger operator $-\Delta+V(x)$ with an external potential $V(x)$, or the pseudo-relativistic counterpart $\sqrt{1-\Delta}$. In these cases our results still apply without significant changes (see Remark \ref{rm:ext}).

Under the assumption (\ref{eq:asp-w}), the Hamiltonian $H_N$ is bounded from below. Therefore, it can be defined as a self-adjoint operator (still denoted by $H_N$) by Friedrichs' method ~\cite{ReeSim2}, and this gives a proper meaning to the many-particle Schr\"odinger equation~\eqref{eq:Schrodinger-dynamics}, by Stone's theorem. 

Bosons have the ability to undergo {\em Bose-Einstein condensation}, which means that a macroscopic number of the particles live in a common quantum state $u(t)\in L^2(\R^d)$. It is a very important fact that (complete) Bose-Einstein condensation is stable under the Schr\"odinger flow in the limit $N\to\ii$. More precisely, if the initial datum is a pure Hartree state 
$$\Psi_{N,0}(x_1,...,x_N)=(u_0)^{\otimes N}(x_1,...,x_N)=u_0(x_1)\cdots u_0(x_N),$$
then the many-particle wave function is well described by a Hartree state
$$\Psi_{N}(t)\approx u(t)^{\otimes N}$$ 
for all times $t\ge 0$ in the limit $N\to\ii$, in a sense to be made precise below (and, in particular, \emph{not} in norm!). Heuristically, the mean-field potential experienced by each particle can be approximated by $|u(t)|^2\ast w$. This observation leads to the nonlinear Hartree equation for the condensate wave function $u(t)$:
\bq \label{eq:Hartree-equation}
\left\{ \begin{gathered}
 i\, \dot u(t) =  \big(-\Delta +|u(t)|^2*w -\mu(t)\big) u(t),  \hfill \\
  u(0)=u_{0}.\hfill \\ 
\end{gathered}  \right.
\eq
Here the gauge parameter $\mu(t)\in \mathbb{R}$ can be freely chosen and different choices lead to different phases of $u(t)$. In this paper it will be convenient to choose 
\bq \label{eq:def-mu}
\mu(t):=\frac12\iint_{\R^d\times\R^d}|u(t,x)|^2w(x-y)|u(t,y)|^2\,dx\,dy
\eq
which implies the compatibility of the energies:
\begin{align*}
N \langle u(t), (-\Delta +|u(t)|^2*w/2) u(t) \rangle \approx
\langle \Psi_{N}(t), H_N \Psi_{N}(t)\rangle = \langle i \dot{\Psi}_{N}(t), \Psi_{N}(t)\rangle \\
\approx N  \langle i \dot{u}(t), u(t) \rangle = N  \langle {u}(t), (-\Delta+|u(t)|^2*w -\mu(t))  u(t) \rangle.
\end{align*}
By the usual argument based on Duhamel's formula, one can show that for every initial datum $u_0\in H^1(\R^d)$, the Hartree equation (\ref{eq:Hartree-equation}) admits a unique solution 
$$u(t)\in C^0([0,\infty),H^1(\R^d))\cap C^1([0,\infty),H^{-1}(\R^d))$$
with the gauge constant $\mu(t)$ given in (\ref{eq:def-mu}). Moreover, the norm $\|u(t)\|_{L^2}$ and the energy $\left \langle u(t), \left( -\Delta + |u(t)|^2*w/2\right) u(t) \right \rangle$ are constant in time. 

In fact, the approximation $\Psi_{N}(t)\approx u(t)^{\otimes N}$ holds in the topology of {\em reduced density matrices}. For every $k\geq1$, the $k$-particle density matrix $\gamma_{\Psi_{N}(t)}^{(k)}$ of $\Psi_{N}(t)$ is the trace class operator on the $k$-particle space $\gH^k=\bigotimes_{\rm sym}^kL^2(\R^d)$ with kernel 
\[\begin{gathered}
\gamma_{\Psi_{N}(t)}^{(k)}(x_1,...,x_k; y_1,...,y_k)\hfill \\
   \quad \quad = \int \Psi_{N}(t,x_1, ... , x_N) \overline{\Psi_{N}(t,y_1,...,y_k, x_{k+1}, ... , x_N)}\, \d x_{k+1}... \d x_{N}.\hfill \\ 
\end{gathered} \]
It is known (see, for example, ~\cite{Spohn-80,ErdYau-01,RodSch-09,KnoPic-10,CheLeeSch-11}) that if $\Psi_N(0)=u(0)^{\otimes N}$, then the $k$-particle density matrix of $\Psi_{N}(t)$ converges to that of the Hartree state $u(t)^{\otimes N}$, that is, 
\bq \label{eq:cv-density-matrices}
\lim_{N\to \infty} \Tr_{\gH^k} \Big| \gamma^{(k)}_{\Psi_{N}(t)}- |u(t) \rangle \langle u(t)|^{\otimes k} \Big| =0,
\eq
for every $k\in \mathbb{N}$ and every $t\geq0$. 

We mention that the convergence (\ref{eq:cv-density-matrices}) does not imply the corresponding approximation in the norm topology of $\gH^N$ for the wave function $\Psi_{N}(t)$. We will see later that even if $\Psi_N(0)=u(0)^{\otimes N}$, then for every $t>0$ fixed, $\Psi_{N}(t)$ {\em never} stays close to the Hartree state $u(t)^{\otimes N}$ in the norm of $\gH^N$, except in the non-interacting case $w\equiv0$. 

The aim of our paper is to derive an effective equation which gives the exact behavior of the wave function $\Psi_{N}(t)$ in the norm topology as $N\to\ii$. More precisely, we shall determine continuous mappings $t\mapsto \phi_k(t)\in \bigotimes_{\rm sym}^k\gH_+(t)$, where $\gH_+(t)=\{u(t)\}^\bot \subset \gH$, such that for all times $t \ge 0$, 
\begin{equation}
\boxed{\lim_{N\to\ii}\norm{\Psi_N(t)-\sum_{k=0}^N u(t)^{\otimes (N-k)}\otimes_s\phi_k(t)}_{\gH^N}=0.}
\label{eq:fluctuations-intro}
\end{equation}
The precise statement is given in Theorem \ref{thm:fluctuations} below. The evolution of the family of functions $(\phi_k(t))_{k=0}^\infty$ is governed by an effective quadratic Hamiltonian in Fock space $\F=\bigoplus_{n\ge 0} \gH^n$, which is derived from Bogoliubov's approximation, as will be explained in detail in Section \ref{sec:Bogoliubov}. The convergence (\ref{eq:fluctuations-intro}) is much more precise than (\ref{eq:cv-density-matrices}) and it implies (\ref{eq:cv-density-matrices}) easily (see Corollary~\ref{cor:Hartree} below).

The behavior of the many-particle wave function in the mean-field regime has been studied in several situations but, as far as we know, never close to a Hartree state $u(t)^{\otimes N}$ in the $N$-body space $\gH^N$ like here. The only exception is \cite{BenKirSch-11}, where the fluctuations around the Hartree evolution of factorized initial data are shown to satisfy a central limit theorem. Most previous works focus instead on the description of the fluctuations around a \emph{coherent state in Fock space}, see for example ~\cite{Hepp-74,GinVel-79,GinVel-79b,GriMacMar-10,GriMacMar-11,Chen-12,GriMac-12,BenOliSch-12}. In the mean-field limit, the evolution of a coherent state is also governed by the {\em same} nonlinear Hartree equation~\eqref{eq:Hartree-equation} in the topology of density matrices, but the Bogoliubov Hamiltonian describing the fluctuations is {\em different} from the Hartree case, as we shall explain in Section~\ref{sec:coherent}. 

Note that, in the time-independent setting, a convergence similar to (\ref{eq:fluctuations-intro}) was recently established for eigenvectors of $H_N$ by Lewin, Nam, Serfaty and Solovej in \cite{LewNamSerSol-13}. In fact, some tools from~\cite{LewNamSerSol-13} will be used in this work. This includes, in particular, the unitary operator $U_N(t)$ defined later and a quantitative estimate in truncated Fock spaces. The convergence (\ref{eq:fluctuations-intro}) will then follows from some energy estimates on the Bogoliubov dynamics and on the Schr\"odinger dynamics.

The paper is organized as follows. In Section \ref{sec:main-result} we state precisely our main result. In Section \ref{sec:coherent} we quickly discuss the fluctuations around coherent states for comparison. Then we derive the Bogoliubov Hamiltonian and study its dynamics in Section \ref{sec:Bogoliubov}. The main result is proved in Section \ref{sec:reformulation}.       
 
\bigskip

\noindent\textbf{Acknowledgement.} We acknowledge financial support from the European Research Council under the European Community's Seventh Framework Programme FP7/2007-2013: Grant Agreement MNIQS 258023 (M.L.~and P.T.N.) and Grant Agreement MAQD 240518 (B.S.). We would also like to thank the referee for carefully reading our manuscript and suggesting some improvement.

\section{Main result} \label{sec:main-result}

\noindent{\bf Notation.} Any Hilbert space we consider has an inner product which is conjugate linear in the first variable and linear in the second. We always denote by $C$ a (large) positive constant which depends only on $w$. 
The symmetric tensor product $\Psi_k\otimes_s\Psi_\ell\in\gH^{k+\ell}$ of two symmetric functions $\Psi_k\in \gH^k$ and $\Psi_\ell\in\gH^\ell$ is defined by
\begin{multline*}
\Psi_k\otimes_s\Psi_\ell(x_1,...,x_{k+\ell})\\
=\frac{1}{\sqrt{k!\ell!(k+\ell)!}}\sum_{\sigma\in\mathfrak{S}_{k+\ell}}\Psi_k(x_{\sigma(1)},...,x_{\sigma(k)})\Psi_\ell(x_{\sigma(k+1)},...,x_{\sigma(k+\ell)}).
\end{multline*}

\begin{theorem}[Dynamics of the fluctuations around a Hartree state]\label{thm:fluctuations}
Assume that $w^2 \le C(1-\Delta)$ on $ \gH=L^2(\R^d)$. Let $u_0\in H^1(\R^d)$ with $\|u_0\|_{L^2}=1$ and let $u(t)$ be the unique solution to the Hartree equation~\eqref{eq:Hartree-equation}. Let $\gH_+(t):=\{u(t)\}^\perp \subset \gH$, and $\gH_+^n(t)=\bigotimes_{\rm sym}^n\gH_+(t)$ be the $n$-fold symmetric tensor product.

Consider a sequence of functions $\phi_{n,0}\in \gH_+^n(0)$ such that, for any $\eps > 0$,  
\begin{equation}
\sum_{n=0}^\ii \norm{\phi_{n,0}}_{\gH^n}^2=1 \quad \text{ and} \quad \sum_{n=0}^\ii \, n \, \langle \phi_{n,0} , (1-\Delta_{x_1}) \phi_{n,0} \rangle < \ii.
\label{eq:condition_initial_datum}  
\end{equation}
If $\Psi_N(t)$ is the solution of the Schr\"odinger equation~\eqref{eq:Schrodinger-dynamics} with initial datum
\begin{equation}
\Psi_N(0)=\sum_{n=0}^N u_0^{\otimes (N-n)}\otimes_s\phi_{n,0},
\label{eq:form_initial_datum} 
\end{equation}
then we have for all times $t\ge 0$, 
\begin{equation}
\lim_{N\to\ii}\norm{\Psi_N(t)-\sum_{n=0}^N u(t)^{\otimes (N-n)}\otimes_s\phi_n(t)}_{\gH^N}=0,
\label{eq:fluctuations}
\end{equation}
with $t\mapsto \phi_n(t)\in\gH_+^n(t)$ a continuous map for all $n\geq0$. More precisely, $\Phi(t):=(\phi_n(t))_{n\geq0}$ is the unique solution to the effective Bogoliubov equation in Fock space $\cF=\bigoplus_{n\geq 0}\gH^n:$
\bq \label{eq:Bogoliubov-dynamics}
\left
\{ \begin{gathered}
  i\, \dot{\Phi}(t) = \bH(t)  \Phi(t) \hfill \\
  \Phi(0)=\Phi_0:=(\phi_{n,0})_{n\geq0}, \hfill \\ 
\end{gathered}  
\right.
\eq
where $\bH(t)$ is a quadratic Hamiltonian in Fock space, written in second quantized form
\begin{align} \label{eq:Bogoliubov-Hamiltonian}
\bH(t)=&\int_{\R^d}a^\dagger(x) \big(-\Delta+|u(t)|^2\ast w-\mu(t) + K_1(t) \big)a(x)\,\d x\\
&+\frac12\iint_{\R^d\times\R^d}\Big(K_2(t,x,y)a^\dagger(x)a^\dagger(y)+\overline{K_2(t,x,y)}a(x)a(y)\Big)\d x\,\d y.\nn
\end{align}
Here $K_1(t)$ is the operator $K_1(t)={Q(t)}\tilde K_1(t)Q(t)$ which acts on $L^2(\R^d)$, where $Q(t)=1-|u(t)\rangle\langle u(t)|$ and where $\tilde{K}_1$ is the Hilbert-Schmidt operator on $L^2(\R^d)$ with kernel $\tilde{K}_1(t,x,y)=u(t,x)w(x-y)\overline{u(t,y)}$. On the other hand, $K_2(t)=(Q(t)\otimes Q(t)) \tilde K_2(t)$ is the projection onto $\gH_+(t)\otimes\gH_+(t)$ of the symmetric two-body function $\tilde K_2(t,x,y)=u(t,x)w(x-y)u(t,y)$. 
\end{theorem}

Here $a^\dagger(x)$ and $a(x)$ are the usual creation and annihilation operators (see Section \ref{sec:coherent} for definitions). The Bogoliubov Hamiltonian $\bH(t)$ and the effective equation \eqref{eq:Bogoliubov-dynamics} will be discussed in detail in Section~\ref{sec:Bogoliubov} below. It is not obvious from the definition of $\bH(t)$ that 
$$\Phi(t)\in \F_+(t):=\bigotimes_{n\ge 0} \gH_+^n(t) \subset \F$$
for all times, but this fact will be proved in the sequel. Note that the Bogoliubov Hamiltonian $\bH(t)$ does {\em not} preserve the subspace $\F_+(t)$ due to the term involving the Hartree mean-field Hamiltonian
$$h(t):=-\Delta+|u(t)|^2\ast w-\mu(t)$$
which itself does not preserve $\gH_+(t)$. In fact, for all times $t$, the vector $\bH(t) \Phi(t)$ belongs to the space tangent to the manifold $\{ \tilde \Phi \in \F_+(t),\ \|\tilde \Phi\|_\F=1\}$ at the point $\Phi(t)$, which is the correct condition to have $\Phi(t)\in\cF_+(t)$.

It is possible to rewrite~\eqref{eq:Bogoliubov-dynamics} as an infinite system of coupled linear equations:
\begin{equation*}
 \left\{\begin{array}{rcl}
i\dot \phi_0(t)&=& \displaystyle\sqrt2\iint_{\R^d\times\R^d}\overline{K_2(t,x_1,x_2)}\phi_2(t,x_1,x_2)\,\d x_1\,\d x_2\\[0.3cm]
i\dot \phi_1(t,x_1)&=& (h(t)+K_1(t))\phi_1(t,x_1)\\[0.3cm]
&&~+\displaystyle\sqrt6\iint_{\R^d\times\R^d}\overline{K_2(t,x_2,x_3)}\phi_3(t,x_1,x_2,x_3)\d x_2\,\d x_3\\[0.3cm]
i\dot \phi_2(t,x_1,x_2)&=& \big((h(t)+K_1(t))_{x_1}+(h(t)+K_1(t))_{x_2}\big)\phi_2(t,x_1,x_2)\\[0.3cm]
&&~+ \displaystyle \sqrt2\, K_2(t,x_1,x_2)\phi_0(t)\\[0.3cm]
&&~+\displaystyle\sqrt{12}\iint_{\R^d\times\R^d}\overline{K_2(t,x_3,x_4)}\phi_4(t,x_1,x_2,x_3,x_4)\d x_3\,\d x_4\\[0.3cm]
&\vdots&
\end{array}
\right.
\end{equation*}
In words, each $\phi_n(t)$ evolves with the particle-preserving Hamilton operator $\sum_{j=1}^n (h(t)+K_1(t))_{x_j}$ and it is coupled to its second left and second right neighbors $\phi_{n-2}(t)$ and $\phi_{n+2}(t)$, through the function $K_2(t,x,y)$. Note that the even functions $\phi_{2j}(t)$ and odd functions $\phi_{2j+1}(t)$ are not coupled. 

The reader should think of the simplest case where $\phi_{0,0}=1$ and $\phi_{n,0}=0$ for all $n\geq1$, which means that the initial datum is the Hartree state 
$$\Psi_N(0)=u_0^{\otimes N}$$ 
for all $N$. In this case the condition~\eqref{eq:condition_initial_datum} boils down to $\|u_0\|_{L^2}=1$ and $u_0\in H^1(\R^d)$. Then~\eqref{eq:fluctuations} gives the exact behavior of the wave function $\Psi_N(t)$ for large $N$, in terms of the ($N$-independent) functions $\phi_n(t)$ which describe the fluctuations around it. Note that in this simple case, it can be verified that the odd functions $\phi_{2j+1}(t)\equiv0$ for all $j$ but the even functions $\phi_{2j}(t)$ are all different from 0 for all $j$ and $t>0$ small enough, provided that $K_2(0)\neq0$. For instance the derivative of $\phi_2$ does not vanish at time zero,
$$\dot\phi_2(0)=- \frac{i}2 K_2(0)\neq 0,$$
and therefore $\phi_2(t)\neq0$ for small times. We deduce from~\eqref{eq:fluctuations} that $\Psi_N(t)$ does \emph{not} stay close to the pure Hartree state $u(t)^{\otimes N}$ in norm in the limit $N\to\ii$. In fact, in the expansion in~\eqref{eq:fluctuations} there are infinitely many terms in the series which do not vanish if $K_2(0)\neq0$, and this is always the case except when $w\equiv 0$.

Finally, we remark that the estimate~\eqref{eq:fluctuations} is much stronger than the convergence of the density matrices. This is because any of the terms $u(t)^{\otimes (N-n)}\otimes_s\phi_n(t)$ in the series has most of its particles condensed into the Hartree state $u(t)$ and only $n$ excited particles outside of the condensate. The functions $\phi_n(t)$ do not participate to the first order of the density matrices in the limit $N\to\ii$. We quickly explain this now for clarity.

\begin{corollary}[Convergence of density matrices]\label{cor:Hartree}
Under the assumptions of Theorem~\ref{thm:fluctuations}, we have
\begin{equation}\label{eq:tr-norm} \lim_{N\to\ii}\tr\left|\gamma^{(k)}_{\Psi_N(t)}-|u(t)\rangle\langle u(t)|^{\otimes k} \right|=0 \end{equation}
for every fixed $k\geq1$ and $t\in\R^+$.
\end{corollary}

\begin{proof}
Recall that for any normalized $\Psi,\Psi'\in\gH^N$
\begin{equation}
\tr\left|\gamma^{(1)}_{\Psi}-\gamma^{(1)}_{\Psi'}\right|\leq 2\norm{\Psi-\Psi'},
\label{eq:estimate_DM} 
\end{equation}
which follows from the estimate
\begin{align*}
\left|\tr\big(A(\gamma^{(1)}_{\Psi}-\gamma^{(1)}_{\Psi'})\big)\right|&=\left|\pscal{\Psi,A\otimes 1_{\gH^{N-1}}\Psi}-\pscal{\Psi',A\otimes 1_{\gH^{N-1}}\Psi'}\right|\\
&\leq 2\norm{A}\,\norm{\Psi-\Psi'}
\end{align*}
for any bounded operator $A$. Combining~\eqref{eq:fluctuations} and~\eqref{eq:estimate_DM} we deduce that
\bq \label{eq:density-cv-1}
\lim_{N\to\ii}\tr\left|\gamma^{(1)}_{\Psi_N(t)}-\gamma^{(1)}_{\sum_{n=0}^N u(t)^{\otimes (N-n)}\otimes_s\phi_n(t)}\right|=0.
\eq
Now we compute, with $P(t)=|u(t)\rangle\langle u(t)|$, 
\begin{align*}
&\pscal{u(t),\gamma^{(1)}_{\sum_{n=0}^N u(t)^{\otimes (N-n)}\otimes_s\phi_n(t)} u(t)}\\
&\qquad=\sum_{n,m=0}^N\pscal{u(t)^{\otimes (N-n)}\otimes_s\phi_n(t),\left(P(t)\otimes 1_{\gH^{N-1}}\right)u(t)^{\otimes (N-m)}\otimes_s\phi_m(t)}\\
&\qquad=\sum_{n=0}^N\left(1-\frac{n}{N}\right)\norm{\phi_n(t)}^2_{\gH^n}
\end{align*}
where we have used the fact that $\phi_n(t)\in\gH_+^N(t)$ for all times. The last term tends to 1 by the dominated convergence theorem for series (actually it will be proved later that $\sum n\|\phi_n(t)\|^2$ stays finite for all times, which even gives the convergence rate $C(t)N^{-1}$). 
It follows that
\bq \label{eq:density-cv-2}
 \lim_{N\to \infty} \tr\, \left| \gamma^{(1)}_{\sum_{n=0}^N u(t)^{\otimes (N-n)}\otimes_s\phi_n(t)} - |u(t)\rangle\langle u(t)| \right| = 0. \eq
From (\ref{eq:density-cv-1}) and (\ref{eq:density-cv-2}) we obtain the convergence of the one-particle density matrices. The result for higher density matrices follows from a similar argument (or follows directly from the convergence of the one-particle density matrices ~\cite[Corollary 2.4]{LewNamRou-13}).
\end{proof}

\begin{remark}[Convergence rates]\label{rmk:rate}
In this work, we do not focus on convergence rates. However, using the techniques that we will discuss in the next sections, one can also obtain a quantitative bound for the difference in (\ref{eq:fluctuations}). The optimal bound is of the order $N^{-1/2}$, for every fixed $t \in \R$, and can be proved proceeding similarly as we do below, with the difference that instead of bounding only the expectation of $\cN$, in Lemma \ref{le:N+-bound} we would need an estimate for the expectation of higher power of $\cN$ (correspondingly, one needs a stronger assumption on the expectation of powers of $\cN$ in the initial data). Such estimates can be established similarly as in the proof of Lemma \ref{le:N+-bound}. Although a bound of the order $N^{-1/2}$ for the fluctuation in (\ref{eq:fluctuations}) immediately implies a bound of the same order for the reduced densities in (\ref{eq:tr-norm}), the optimal estimate for the latter is proportional to $N^{-1}$ (see \cite{CheLeeSch-11}).
\end{remark}

\begin{remark}[Extensions] \label{rm:ext} We can consider not only the typical Hamiltonian $H_N$ in (\ref{eq:def_H_N}) but also many other cases.
\begin{itemize}

\item[(i)] (External potential) We can add to the Laplacian $-\Delta$ in $H_N$ an external potential $V(x)$ which satisfies 
$$V \in L_{\rm loc}^{d/2} (\R^d,\R) \quad {\text{and}}~~-V(x) \le \alpha (-\Delta)+C$$
for some $0\le \alpha<1$. Note that $V(x)$ is allowed to grow at infinity, which corresponds to a confined system. The potential $V(x)$ appears in the Bogoliubov Hamiltonian in a natural way. Our results in Theorem \ref{thm:fluctuations} still hold true provided that $w^2 \le C(1-\Delta)$ and that the second condition in (\ref{eq:condition_initial_datum}) on the initial datum is replaced by   
$$
\sum_{n=0}^\ii \, n \, \langle \phi_{n,0} , (1-\Delta_{x_1}+|V(x_1)|) \phi_{n,0} \rangle < \ii.
$$

\item[(ii)] (Boson stars) We can also consider bosons living in $\R^3$ with the one-body pseudo-relativistic kinetic operator $\sqrt{1-\Delta}$ and the Newton interaction $w(x-y)=-\kappa |x-y|^{-1}$, $0<\kappa<2/\pi$. If $u_0\in H^1(\R^3)$, then the Hartree dynamics $u(t)$ is well-posed in $H^1(\R^3)$ (see \cite{Lenzmann-07}) and our results in Theorem \ref{thm:fluctuations} still hold true under the same $H^1$-condition (\ref{eq:condition_initial_datum}) on the initial datum.
\end{itemize}
\end{remark}

\section{Comparison with coherent states}\label{sec:coherent}

One common approach to derive Hartree's nonlinear equation from the many-body dynamics is to work in Fock space and to use coherent states. We quickly recall this method here and we compare it with our main result. The Fock space is
$$\cF=\C\oplus\bigoplus_{n\geq1}\gH^n$$
and it is useful to introduce the creation and annihilation operators $a^\dagger(f)$ and $a(g)$ which are unbounded operators acting on $\cF$. More precisely, $a^\dagger(f)$ is defined by
$$a^\dagger(f)\left(\phi_0\oplus\bigoplus_{n\geq1} \phi_n\right)=0\oplus \phi_0 f \oplus\bigoplus_{n\geq1} \big(f\otimes_s \phi_n\big)$$
and $a(f)$ is its adjoint. The operator-valued distributions $a^\dagger(x)$ and $a(x)$ used in (\ref{eq:Bogoliubov-Hamiltonian}) may be defined so that  
\bq \label{eq:af-ax}
 a^\dagger (f)=\int _{\R^d} f(x)a^\dagger(x) \d x \quad {\rm and}\quad a(f)=\int _{\R^d} \overline{f(x)}a^\dagger(x) \d x
 \eq
for all $f\in \gH$. These operators satisfy the canonical commutation relations
\begin{align} 
\label{eq:CCR-af-ag}
[a(f),g(g)] &=[a^\dagger(f),a^\dagger(g)]=0,\quad [a(f),a^\dagger(g)]=\langle f,g\rangle_{\gH} \\
\label{eq:CCR_ax-ay}
[a(x),a(y)]&=[a^\dagger(x),a^\dagger(y)]=0,\quad [a(x),a^\dagger(y)]=\delta(x-y)
\end{align}
for all $f,g\in \gH$ and for all $x,y\in \R^d$. 

The \emph{Weyl unitary operator} is defined by
$$W(f):=\exp\big(a^\dagger(f)-a(f)\big)$$
and a \emph{coherent state} is a Weyl-rotation of the vacuum $\Omega=1\oplus0\cdots$:
$$W(f)\Omega=e^{-\norm{f}^2/2}\sum_{n\geq0}\frac{1}{\sqrt{n!}}f^{\otimes n}.$$
The average particle number of a coherent state is, therefore,
$$e^{-\norm{f}^2}\sum_{n\geq0}\frac{n}{n!}\norm{f}^{2 n}=\norm{f}^2.$$

Instead of starting the many-particle dynamics in the neighborhood of a Hartree state in $\gH^N$, one can also use a coherent state in Fock space. The analysis is simplified by the useful algebraic properties of coherent states. For example, coherent states are eigenvectors of annihilation operators, which follows from the action of the Weyl rotation $W(f)$ on the creation and annihilation operators  
\bq \label{eq:Weyl-action}
W(f)^* a^\dagger(g) W(f)=a^\dagger(g) + \langle f,g \rangle, \quad W(f)^* a(g) W(f)=a(g) + \langle g,f \rangle.
\eq

For a precise discussion on the dynamics in Fock space, let us introduce the second quantized forms of operators. If $A$ is a one-body operator on $\gH$ with kernel $A(x,y)$, then its second quantized form $d\Gamma(A)$ is an operator on $\F$ defined by  
\bq \label{eq:2nd-dG1}
d\Gamma(A):= \bigoplus_{n\ge 1} \sum_{j=1}^n A_j= \iint_{\R^d\times \R^d} a^\dagger(x) A(x,y) a(y) \,\d x\,\d y,
\eq
which is frequently written as $\int a^\dagger(x) A a(x) \,dx$ for short. For example, $\N:=d\Gamma(\1_\gH)$ is the number operator. Similarly, if $W$ is a symmetric operator on $\gH^2$ with kernel $W(x,y;x',y')$, then its second quantized form is given by
\begin{align}\label{eq:2nd-dG2}
&\bigoplus_{n\ge 2} \bigg( \sum_{1\le j<k \le n} W_{j,k} \bigg)\\
&\quad \quad = \frac{1}{2}\iiiint W(x,y;x',y')a^\dagger(x)a^\dagger(y)a(x') a(y') \,\d x\,\d y\,\d x'\,\d y'. \nn
\end{align}
In particular, the formal kernel of the multiplication operator $w(x-y)$ on $\gH^2$ is $w(x-y)\delta_{(x,y)}(x',y')$ and the second quantized form of $H_N$ is an operator on Fock space:
\begin{equation}
\bH_N=d\Gamma(-\Delta) +\frac{1}{2(N-1)}\iint_{\R^d\times\R^d} w(x-y)a^\dagger (x)a^\dagger (y) a(x) a(y) \, \d x \, \d y \label{eq:bH_N}\\
\end{equation}
which coincides with $H_N$ in the sector $\gH^N$. Because of the constant in front of the interaction which depends on $N$ and not on the number operator $\cN$, the restriction of $\bH_N$ to another $k$-particle subspaces is not related to $H_k$.
Under certain assumptions on $w$, for example (\ref{eq:asp-w}), the dynamics generated by $\bH_N$ is well-defined on the quadratic form domain of $d\Gamma(1-\Delta)$.

It was first proved by Hepp \cite{Hepp-74} for regular interactions, and then extended by Ginibre-Velo \cite{GinVel-79,GinVel-79b} for singular interactions (see also~\cite{RodSch-09,GriMacMar-10,GriMacMar-11}) that if the initial state (in Fock space) is close to a coherent state $W(\sqrt{N} u_0)\Omega$, then the dynamics generated by $\bH_N$ stays close to $W(\sqrt{N} u(t))\Omega$ for all $t\ge 0$ in the sense of density matrices and it is also possible to identify the fluctuations. In our typical setting, this  result can be stated as follows.

\begin{theorem}[Dynamics of the fluctuations around a coherent state]\label{thm:fluctuations_coherent}
Assume that $w^2 \le C(1-\Delta)$. Let $u_0\in H^1(\R^d)$ with $\|u_0\|_{L^2}=1$ and let $u(t)$ be the unique solution to the nonlinear Hartree equation~\eqref{eq:Hartree-equation}.

Consider an initial datum $\Psi_{N,0}$ in Fock space which is such that 
$$W(\sqrt{N} u_0)^*\Psi_{N,0} \to \Xi_0$$
strongly in $\F$ and weakly in $Q(\dGamma(1-\Delta))$. Let $\Psi_N(t)=\exp(-it\bH_N)\Psi_{N,0}$ be the solution of the Schr\"odinger equation in $\cF$, with initial datum $\Psi_{N,0}$. Then we have for all times $t\ge 0$:
\begin{equation}
\lim_{N\to\ii}W(\sqrt{N} u(t))^*\Psi_N(t) = \Xi(t),
\label{eq:fluctuations_coherent}
\end{equation}
strongly in $\cF$ and weakly in $Q(\dGamma(1-\Delta))$, where $\Xi(t)$ is the unique solution to the effective equation in Fock space
\bq \label{eq:Bogoliubov-dynamics_coherent}
\left
\{ \begin{gathered}
  i\, \dot{\Xi}(t) = \tilde\bH(t)  \Xi(t) \hfill \\
  \Xi(0)=\Xi_0. \hfill \\ 
\end{gathered}  
\right.
\eq
Here $\tilde \bH(t)$ is a quadratic Hamiltonian on $\cF$, 
 written in second quantized form
\begin{align*}
\tilde\bH(t)=&\int_{\R^d}a^\dagger(x) \big( -\Delta +|u(t)|^2\ast w-\mu(t) + \tilde K_1(t) \big)a(x)\,\d x\nn\\
&+\frac12\iint_{\R^d\times\R^d}\Big(\tilde K_2(t,x,y)a^\dagger(x)a^\dagger(y)+\overline{\tilde K_2(t,x,y)}a(x)a(y)\Big)\d x\,\d y,
\end{align*}
with $\tilde K_1(t,x,y)=u(t,x)w(x-y)\overline{u(t,y)}$, $\tilde K_2(t,x,y)=u(t,x)w(x-y)u(t,y)$.
\end{theorem}

Theorem \ref{thm:fluctuations_coherent} can be proved using the argument of the proof of Theorem \ref{thm:fluctuations} and we will omit the details. Although the result in the case of coherent states looks very similar to the Hartree case, the effective Bogoliubov Hamiltonians $\tilde\bH(t)$ and $\bH(t)$ are \emph{not} the same. The part involving $h(t)$ is identical, but the other terms involve the functions $\tilde K_j(t,x,y)$'s instead of the projected ones $K_j(t,x,y)$'s. 

In fact, the coherent state approach can also be used to study fluctuations around a Hartree state, if the initial data is exactly a Hartree state. This can be seen as follows. Theorem~\ref{thm:fluctuations_coherent} implies that
$$\lim_{N\to\ii}\norm{e^{it\bH_N}W(\sqrt{N}u_0)\Omega-W(\sqrt{N}u(t))\tilde{\mathbb{U}}(t,0)\Omega}_{\cF}=0,$$
where $\tilde{\mathbb{U}}(t,0)$ is the unitary propagator associated with the time-dependent Bogoliubov Hamiltonian $\tilde\bH(t)$ in Fock space. A more precise analysis shows that 
\[ \norm{e^{it\bH_N}W(\sqrt{N}u_0)\Omega-W(\sqrt{N}u(t))\tilde{\mathbb{U}}(t,0)\Omega}_{\cF} \lesssim N^{-1/2} \]
for any fixed time $t$. Hence projecting onto the $N$-particle subspace, we find
\begin{equation}
\norm{\frac{(N/e)^{N/2}}{\sqrt{N!}}e^{itH_N}u_0^{\otimes N}-\1_{\gH^N}W(\sqrt{N}u(t))\tilde{\mathbb{U}}(t,0)\Omega}_{\gH^N} \lesssim N^{-1/2} \, .
\label{eq:projection_coherent}
\end{equation}
Multiplying with $\sqrt{N!} / (N/e)^{N/2} \simeq N^{1/4}$, we get  
$$ \lim_{N \to \infty} \norm{e^{itH_N}u_0^{\otimes N}-(2\pi N)^{1/4}\1_{\gH^N}W(\sqrt{N}u(t))\tilde{\mathbb{U}}(t,0)\Omega}_{\gH^N}= 0.$$
The effective dynamics $(2\pi N)^{1/4}\1_{\gH^N}W(\sqrt{N}u(t))\tilde{\mathbb{U}}(t,0)\Omega$ agrees with our fluctuations in~\eqref{eq:fluctuations} to the leading order, even if it looks much more complicated. In fact, this analysis can be extended to initial $N$-particle wave functions 
having the form $d_N \, \1_{\gH^N} W(\sqrt{N} u_0) \Psi$, assuming the normalization constant $d_N$ to be small enough (of order $N^{1/4}$) and $\Psi$ to have sufficiently small number of particles (of order one, independent of $N$). In other words, to get information about the evolution of an $N$-particle initial data using the coherent state method requires further assumptions on the initial datum, which guarantee sufficient closeness to an Hartree state (in an appropriate sense). On the other hand, our direct method in $\gH^N$ applies to any initial datum of the general form~\eqref{eq:form_initial_datum} satisfying~\eqref{eq:condition_initial_datum}, for which no convergence rate is known.

\section{The Bogoliubov dynamics}\label{sec:Bogoliubov}

\subsection{Describing fluctuations around a given $u(t)$} \label{subsec:UN}
Consider a time-dependent normalized vector $u(t)$ in $\gH$ (which does not necessarily satisfy the Hartree equation). As in~\cite[Sec. 2.3]{LewNamSerSol-13}, we can write any function $\Psi\in \gH^N$ as follows
\begin{equation}
\Psi=\psi_0\, u(t)^{\otimes N}+ \psi_1 \otimes_s u(t)^{\otimes (N-1)} +\psi_2 \otimes_s u(t)^{\otimes (N-2)}+\cdots + \psi_N
\label{eq:decomp_Psi} 
\end{equation}
where $\psi_n\in \gH_+(t)^{n}:= \bigotimes_{\rm sym}^n \gH_+(t)$ with $\gH_+(t):= \{u(t)\}^{\bot}$. Following~\cite{LewNamSerSol-13}, we define the unitary
\begin{equation}
 \label{eq:def-unitary-UN}
\begin{array}{cccl}
U_{N}(t): & \gH^N & \to & \displaystyle\cF_+(t)^{\leq N}=\bigoplus_{n=0}^N \gH_+(t)^n \\[0.3cm]
 & \Psi & \mapsto & \psi_0\oplus \psi_1 \oplus\cdots \oplus \psi_N.
\end{array}
\end{equation}
Note that we have the inclusions
\begin{equation}
\cF_+(t)^{\leq N}\subset \cF_+(t) = \bigoplus_{n=0}^\infty \gH_+(t)^n \subset \cF.
\label{eq:inclusions} 
\end{equation}
Therefore, we can always see $U_{N}(t)$ as a partial isometry from $\gH^N$ to $\cF$, where $U_{N}(t)^*$ is extended by $0$ outside of $\cF_+(t)^{\leq N}$. 

Using $a(u(t)) \psi_n=0$, $Q(t) a^\dagger(u(t))=0$, $[a(u(t)),a^\dagger(u(t))]=1$ and
$$\psi_n \otimes_s u(t)^{\otimes (N-n)} = \frac{(a^\dagger(u(t)))^{N-n}}{\sqrt{(N-n)!}} \psi_n$$
we find that  
$$
Q(t)^{\otimes j} \frac{a (u(t))^{N-j}}{\sqrt{(N-j)!}} \left( \psi_n \otimes_s u(t)^{\otimes (N-n)} \right)=\begin{cases}
\psi_n & \text{if $j=n$,}\\[0.1cm]
0 & \text{if $j \ne n$.}\\
\end{cases}
$$
Therefore, the operators $U_{N}(t)$ and $U_{N}(t)^*$ can be equivalently written as
\bq \label{eq:alt-def-UN}
\begin{split}
U_{N}(t)\Psi &= \bigoplus_{j=0}^N Q(t)^{\otimes j} \frac{a (u(t))^{N-j}}{\sqrt{(N-j)!}}\Psi ,\\ U_{N}(t)^* \left( \bigoplus_{j=0}^N \psi_j \right) &= \sum_{j = 0}^N\frac{{{{a^\dagger (u(t))}^{N - j}}}}{{\sqrt {(N - j)!} }} \psi_j
\end{split}
\eq
for all $\Psi\in \gH^N$ and $\psi_j\in \gH_+(t)^j$, where $Q(t)=1-|u(t)\rangle\langle u(t)|$, and $a(u(t))$ and $a^\dagger(u(t))$ are the usual annihilation and creation operators on $\F$. Moreover, as shown in ~\cite[Proposition 14]{LewNamSerSol-13}, we have the following identities on $\F_+(t)^{\le N}$:
\bq 
\label{eq:UN-actions-1}
U_{N}(t) \, a^\dagger (u(t)) a(u(t)) \,U_{N}(t)^* &=& N- \N_+(t) , \hfill \\
\label{eq:UN-actions-2}
U_{N}(t)\, a^\dagger(f) a(u(t)) \,U_{N}(t)^* &=& a^\dagger(f) \sqrt{N-\N_+(t)},\hfill\\
\label{eq:UN-actions-3}
U_{N}(t) \,a^\dagger (u(t)) a(f) \,U_{N}(t)^* &=& \sqrt{N-\N_+(t)}\, a(f),\hfill\\
\label{eq:UN-actions-4}
U_{N}(t)\, a^\dagger(f) a(g) \,U_{N}(t)^* &=& a^\dagger(f) a(g),
\eq
for all $f,g\in \gH_+(t)$, where $\N_+(t)=\N-a^\dagger(u(t)) a(u(t))$ is the number operator on $\F_+(t)$. For the reader's convenience, let us quickly verify these useful formulas. First, when $f,g\in \gH_+(t)$, since $a^\dagger(f)a(g)$ commutes with $a^\dagger(u(t))$, from \eqref{eq:alt-def-UN} we deduce that 
\begin{align*}
U_N(t) a^\dagger(f) a(g) &U_N(t)^* \left( \bigoplus_{j=0}^N \psi_j \right) \\ &= U_N(t) a^\dagger(f) a(g) \left( \sum_{j=0}^N \frac{a^\dagger(u(t))^{(N-j)}}{\sqrt{(N-j)!}} \psi_j\right) \\
& =  U_N(t)  \left( \sum_{j=0}^N \frac{a^\dagger(u(t))^{(N-j)}}{\sqrt{(N-j)!}} \Big(a^\dagger(f) a(g) \psi_j \Big)\right) \\
&= \bigoplus_{j=0}^N \Big(a^\dagger(f) a(g) \psi_j \Big) = a^\dagger(f) a(g) \left( \bigoplus_{j=0}^N \psi_j \right) 
\end{align*}
for all $\psi_j\in \gH_+(t)^j$. Thus \eqref{eq:UN-actions-4} holds true. A consequence of \eqref{eq:UN-actions-4} is  
$$U_N(t) \cN_+(t) U_N(t)^*= \cN_+(t)$$
which is equivalent to \eqref{eq:UN-actions-1}. Finally, when $f\in \gH_+(t)$, using the fact that $a(f)$ commutes with $a^\dagger(u(t))$ and \eqref{eq:alt-def-UN}, we have 
\begin{align*}
U_N(t) a^\dagger(u(t)) & a(f) U_N(t)^* \left( \bigoplus_{j=0}^N \psi_j \right) \\
&= U_N(t) a^\dagger(u(t)) a(f) \left( \sum_{j=0}^N \frac{a^\dagger(u(t))^{(N-j)}}{\sqrt{(N-j)!}} \psi_j\right) \\
& =  U_N(t)  \left( \sum_{j=0}^N \frac{a^\dagger(u(t))^{(N-j+1)}}{\sqrt{(N-j+1)!}} \sqrt{N-j+1} \Big( a(f) \psi_j \Big)\right) \\
&= \bigoplus_{j=0}^N  \sqrt{N-j+1} \Big( a(f) \psi_j \Big) = \sqrt{N-\cN_+(t)} a(f)  \left( \bigoplus_{j=0}^N \psi_j \right)
\end{align*}
for all $\psi_j\in \gH_+(t)^j$, which verifies \eqref{eq:UN-actions-3}. The identity \eqref{eq:UN-actions-2} can be proved by the same way.

In the analysis of the dynamics of $U_{N}(t) \Psi_{N}(t)$, not only $U_{N}(t)$ but also its time-derivative plays an important role. We have
\begin{lemma} Let $u(t)$ be an arbitrary (sufficiently regular) trajectory on $\gH$ satisfying $\|u(t)\|=\|u(0)\|$ for all $t\ge 0$. Then the time-derivative of $U_N(t)$ is
\begin{multline} \label{eq:derivative-UN}
i\dot{U}_{N}(t)=\Big( a^\dagger(u(t))a(Q(t)i \dot u(t))-\sqrt{N-\N_+(t)}\,a(Q(t)i\dot u(t))\\
-a^\dagger(Q(t)i\dot u(t))\sqrt{N-\N_+(t)}-\pscal{i\dot u(t),u(t)}(N-\N_+(t))\Big)U_{N}(t).
\end{multline}
\end{lemma}
Note that the last three terms in the right side of (\ref{eq:derivative-UN}) 
do not go outside of $\cF_+(t)$, but the first does. 

\begin{proof} We start by writing
\begin{equation}
U_{N}(t)= \bigoplus_{k=0}^N \; Q(t)^{\otimes k}\frac{a(u(t))^{N-k}}{\sqrt{(N-k)!}}=:\bigoplus_{k=0}^N \cL_k(t).
\label{eq:formula-psi-k-a}
\end{equation}
Now taking the time-derivative, we get
\begin{align} \label{eq:derivative-Lk}
\dot\cL_k(t) &=\frac{N-k}{\sqrt{(N-k)!}}Q(t)^{\otimes k} a(\dot{u}(t))a(u(t))^{N-k-1} \nn\\
&~+ \frac{1}{\sqrt{(N-k)!}}\;\Big(\dot{Q}(t)\otimes Q(t)\otimes\cdots\otimes Q(t)+\cdots\nn\\
& \qquad\qquad\qquad\qquad\quad + Q(t)\otimes\cdots\otimes Q(t)\otimes \dot{Q}(t)\Big)a(u(t))^{N-k}.
\end{align}
In the first term of the right side of (\ref{eq:derivative-Lk}), we have used that $a(\dot{u}(t))$ and $a(u(t))$ commute. It is clear that for any $v$, we have
$$Q(t)^{\otimes (n+1)} a^\dagger(v)_{|\gH^{n}}={a^\dagger(Q(t)v)Q(t)^{\otimes n}}_{|\gH^{n}}$$
and
$$Q(t)^{\otimes n} a(v)_{|\gH^{n+1}}={a(Q(t)v)Q(t)^{\otimes (n+1)}}_{|\gH^{n+1}}+\pscal{v,u(t)}Q(t)^{\otimes n} a(u(t))_{|\gH^{n+1}}.$$
So the first term of the right side of (\ref{eq:derivative-Lk}) can be rewritten as 
\begin{align} \label{eq:derivative-Lk-a}
&\frac{N-k}{\sqrt{(N-k)!}}Q(t)^{\otimes k}\Big( a(\dot{u}(t))a(u(t))^{N-k-1}\Psi\Big)\nn\\
&=\frac{\sqrt{N-k}}{\sqrt{(N-k-1)!}} a(Q(t)\dot{u}(t)) Q(t)^{\otimes (k+1)}a(u(t))^{N-k-1}\nn\\
&\qquad+\pscal{\dot{u}(t),u(t)}\frac{N-k}{\sqrt{(N-k)!}}Q(t)^{\otimes k} a(u(t))^{N-k}\nn\\
&=\sqrt{N-k}\, a(Q(t)\dot{u}(t)) \cL_{k+1}(t)+\pscal{\dot{u}(t),u(t)}(N-k)\cL_k(t).
\end{align}
Note that 
\begin{align*}
\dot{Q}(t)=-|u(t)\rangle\langle \dot{u}(t)|-|\dot{u}(t)\rangle\langle {u}(t)|=-|u(t)\rangle\langle \dot{u}(t)|Q(t)-Q(t)|\dot u(t)\rangle\langle {u}(t)|
\end{align*}
because 
\begin{align*}
|u(t)\rangle\langle \dot{u}(t)|&=|u(t)\rangle\langle \dot{u}(t)|Q(t)+\pscal{\dot{u}(t),u(t)}P(t)\\
|\dot u(t)\rangle\langle {u}(t)|&=Q(t)|\dot u(t)\rangle\langle {u}(t)|+\pscal{u(t),\dot u(t)}P(t),
\end{align*}
where $P(t)=|u(t)\rangle \langle u(t)|$, and 
$$\pscal{\dot{u}(t),u(t)}+\pscal{u(t),\dot u(t)}=\frac{d}{dt} ||u(t)||^2=0.$$
Since $a^\dagger(f)a(g)=\dGamma(|f\rangle\langle g|)$, we find
\begin{align*}
&\dot{Q}(t)\otimes Q(t)\otimes\cdots\otimes Q(t)+\cdots+ Q(t)\otimes\cdots\otimes Q(t)\otimes \dot{Q}(t)\\
&\qquad=-\left(\sum_{j=1}^k|u(t)\rangle\langle Q(t)\dot{u}(t)|_j\right)Q(t)^{\otimes k}-Q(t)^{\otimes k}\left(\sum_{j=1}^k|Q(t)\dot u(t)\rangle\langle {u}(t)|_j\right)\\
&\qquad=-a^\dagger(u(t))a(Q(t)\dot u(t))Q(t)^{\otimes k}-Q(t)^{\otimes k}a^\dagger(Q(t)\dot u(t))a(u(t))\\
&\qquad=-a^\dagger(u(t))a(Q(t)\dot u(t))Q(t)^{\otimes k}-a^\dagger(Q(t)\dot u(t))Q(t)^{\otimes (k-1)}a(u(t)).
\end{align*}
So the second term of the right side of (\ref{eq:derivative-Lk}) can be rewritten as 
\begin{flalign} \label{eq:derivative-Lk-b}
&\frac{1}{\sqrt{(N-k)!}}\;\Big(\dot{Q}(t)\otimes Q(t)\otimes\cdots\otimes Q(t)+\cdots \nn\\
&\qquad\qquad\qquad\qquad\qquad+ Q(t)\otimes\cdots\otimes Q(t)\otimes \dot{Q}(t)\Big)a(u(t))^{N-k}\nn\\
&=-a^\dagger(u(t))a(Q(t)\dot u(t))\cL_k-\sqrt{N-k+1}\,a^\dagger(Q(t)\dot u (t))\cL_{k-1}.
\end{flalign}

Substituting (\ref{eq:derivative-Lk-a}) and (\ref{eq:derivative-Lk-b}) into (\ref{eq:derivative-Lk}) and taking the sum over $k$, we deduce from (\ref{eq:formula-psi-k-a}) the formula
\begin{multline*}
\dot{U}_N(t)=\Big(-a^\dagger(u(t))a(Q(t)\dot u(t))+\sqrt{N-\N_+(t)}\,a(Q(t)\dot u(t))\\
-a^\dagger(Q(t)\dot u(t))\sqrt{N-\N_+(t)}+\pscal{\dot u(t),u(t)}(N-\N_+(t))\Big)U_N(t),
\end{multline*}
which is equivalent to (\ref{eq:derivative-UN}). 
\end{proof}

\subsection{Derivation of the effective evolution} \label{subsec:Bogoliubov-dynamics} Let $\Psi_{N}(t)$ satisfy the Schr\"odinger equation (\ref{eq:Schrodinger-dynamics}). Let $u(t)$ be the solution to the nonlinear Hartree equation and let $U_{N}(t)$ be as in (\ref{eq:def-unitary-UN}). Let us consider the vector
$$\boxed{\Phi_{N}(t):= U_{N}(t) \Psi_{N}(t)}$$
in the Fock space $\F_+^{\le N}(t) \subset \F_+(t)\subset \F$. From (\ref{eq:Schrodinger-dynamics}) and (\ref{eq:derivative-UN}), we get
\begin{multline} \label{eq:derivative-Phi-Nt}
i\dot{\Phi}_{N}(t)=\Big( U_{N}(t) H_N U_{N}(t)^* + a^\dagger(u(t))a(Q(t)i\dot u(t))-\pscal{i\dot u(t),u(t)}(N-\N_+(t))\\ -\sqrt{N-\N_+(t)}\,a(Q(t)i\dot u(t))-a^\dagger(Q(t)i\dot u(t))\sqrt{N-\N_+(t)}  \Big) \Phi_{N}(t).
\end{multline} 
A lengthy but straightforward computation as in \cite[Section 4]{LewNamSerSol-13} shows that 
\begin{flalign} \label{eq:UNHNUN-original}
& U_{N}(t)H_{N} U_{N}(t)^* = N e(t) + \dGamma \Big( Q(t) \big(h(t)+K_1(t)-e(t)\big) Q(t) \Big)\nn\\
&\quad+ \sqrt{N-\N_+(t)} a \big( Q(t)h(t) u(t)\big)+a ^*\big( Q(t) h(t)u(t)\big)\sqrt{N-\N_+(t)} \nn\\
&\quad+\frac12\iint_{\R^d\times\R^d}\bigg( K_2(t,x,y)a^\dagger(x)a^\dagger(y)+\overline{K_2(t,x,y)}a(x)a(y)\bigg)\d x\,\d y \nn\\
&\quad + R_{1,N}(t) + R_{2,N}(t)
\end{flalign} 
where 
\begin{align} \label{eq:RN1}
R&_{1,N} (t)= \frac{1}{2} \dGamma( Q(t) [w*|u(t)|^2+K_{1}(t)-\mu(t)]Q(t))\frac{1-\N_+(t)}{N-1} \nn\\
&- \frac{\N_+(t) \sqrt{N-\N_+(t)}}{N-1} a \big(Q(t) [w*|u(t)|^2]u(t)\big) \nn\\
&+  \frac{1}{2} \iint  K_2(t,x,y)a^\dagger(x) a^\dagger(y) \d x\, \d y\, \bigg( \frac{\sqrt{(N-\N_+(t))(N-\N_+(t)-1)}}{N-1} -1 \bigg)\nn \\
& + \frac{\sqrt{N-\N_+(t)}}{N-1} \iiiint  \, \big( \1 \otimes Q(t) \, w \, 
Q(t) \otimes Q(t) \big) (x,y; x', y') \nn\\ & \quad\quad\quad\quad\quad\quad\quad\quad\quad\quad\quad\quad\quad\quad \times u (t,x) \, a^\dagger(y) a(x')a(y')\,\d x\, \d y\, \d x'\, \d y' \nn \\
&+\text{h.c.}
\end{align}
(we have written $X+{\rm h.c.}$ instead of $X+X^*$ for short) and
\begin{align} \label{eq:RN2} 
R_{2,N} (t) =  \frac{1}{2(N-1)} \iiiint & \big( Q(t)\otimes Q(t) \, w \, Q(t) \otimes Q(t) \big) (x,y; x',y') \nn \\ & \times a^\dagger(x) a^\dagger(y) a(x')a(y')\,\d x\, \d y\, \d x'\, \d y' .
\end{align}
Here $w$ is always understood as the multiplication operator by the function $w(x-y)$ in $\gH^2$ and recall that 
\begin{align*}
P(t) & =1-Q(t)=|u(t) \rangle \langle u(t)|,\\
h(t) &=-\Delta+|u(t)|^2*w-\mu(t)\\
e(t) &=\langle u(t), h(t) u(t) \rangle = \langle u(t), (-\Delta + |u(t)|^2*w/2) u(t) \rangle.
\end{align*}
The idea behind (\ref{eq:UNHNUN-original}) is that, in accordance with (\ref{eq:UN-actions-1})-(\ref{eq:UN-actions-4}), the unitary operator $U_N (t)$ implements a c-number substitution; it maps $a (u(t))$ and $a^* (u(t))$ to $\sqrt{N- \cN_+ (t)}$ (which in first approximation can be replaced by the number $\sqrt{N}$), leaving instead creation and annihilation operators with arguments orthogonal to $u(t)$ invariant. Writing $H_N$ in second quantized form (similarly to (\ref{eq:bH_N}), although here only the action of $H_N$ on the $N$-particle sector matters) and decomposing the arguments of creation and annihilation operators in components proportional and orthogonal to $u(t)$, we end up with (\ref{eq:UNHNUN-original}); the details can be found in Appendix \ref{app:UN-HN-UN*}.

Now we substitute (\ref{eq:UNHNUN-original}) into (\ref{eq:derivative-Phi-Nt}) and use the Hartree equation $i\dot u(t) =h(t) u(t)$. Since 
$$
 a^\dagger(u(t))a(Q(t) h(t) u(t))+d\Gamma( Q(t) h(t) Q(t))= \dGamma( h(t) ) - \dGamma( h(t) P(t)) ,
$$
and $\Phi_N(t)\in \F_+(t)$, we find that 
\bq \label{eq:equation_Phi_N_t}
\boxed{i\dot{\Phi}_{N}(t) =  \Big(  \bH(t) + R_{1,N}(t)+R_{2,N}(t) \Big)\Phi_{N}(t)}
\eq
where $\bH(t)$ is the Bogoliubov Hamiltonian given in (\ref{eq:Bogoliubov-Hamiltonian}):
\begin{align*} 
\bH(t)= d\Gamma(h(t)+K_1(t))+ \frac12\iint\Big(  K_2(t,x,y)&a^\dagger(x)a^\dagger(y)\\
&+\overline{K_2(t,x,y)}a(x)a(y)\Big)\, dx\,dy.
\end{align*}

Heuristically, if most particles live in the state $u(t)$, then $\N_+(t)/N$ is negligible and hence  $R_{1,N}(t)+R_{2,N}(t)$ can be ignored. The main goal of our paper is to show that ${\Phi}_{N}(t)$ actually converges to the exact solution of the Bogoliubov equation   
\bq \label{eq:formal-Bogoliubov-equation}
\boxed{i\dot{\Phi}(t) = \bH(t) \Phi(t).}
\eq

The convergence $\Phi_N(t)\to \Phi(t)$ will be reformulated and justified in Section \ref{sec:reformulation}. In the rest of this section, we shall study the well-posedness of equation (\ref{eq:formal-Bogoliubov-equation}). 

\subsection{Well-posedness of the Bogoliubov dynamics}\label{sec:Dynamics-quadratic-form} Now we state and prove a result concerning the well-posedness of the Bogoliubov dynamics (\ref{eq:formal-Bogoliubov-equation}). The main difficulty is of course the fact that $\bH(t)$ is time-dependent.

\begin{theorem}[Bogoliubov dynamics] \label{thm:Bogoliubov-dynamics} Assume that $w^2 \le C(1-\Delta)$. Let $u_0\in H^1(\R^d)$ with $\|u_0\|=1$ and let $u(t)$ be the unique solution to the Hartree equation~\eqref{eq:Hartree-equation}. For every state $\Phi_0$ in the quadratic form domain $Q(d\Gamma(1-\Delta))$, there exists a unique solution to the Bogoliubov equation
\bq \label{eq:Bog-eq}
\left\{ \begin{gathered}
   i \dot\Phi(t)=\bH(t) \Phi(t),\hfill \\
   \Phi(0)=\Phi_0\hfill \\ 
\end{gathered}  \right.
\eq
such that $\Phi\in C^0([0,\ii),\cF)\cap L^\ii_{\rm loc}(0,\ii;Q(d\Gamma(1-\Delta)))$.
Moreover, $\Phi(t)\in \F_+(t)$ and 
$$\langle \Phi(t), d\Gamma(1-\Delta) \Phi(t) \rangle \le C e^{Ct} \langle \Phi_0, d\Gamma(1-\Delta) \Phi_0 \rangle .$$ 
\end{theorem}

The proof of Theorem \ref{thm:Bogoliubov-dynamics} is based on the following abstract result, which is inspired by works of Kisy\'nski \cite{Kis-64} and Simon~\cite{Simon-71}.  

\begin{theorem}[Dynamics generated by quadratic forms]\label{thm:dynamics-unbounded} Let $A\ge 1$ be a self-adjoint operator and $\{H(t)\}_{t\in [0,1)}$ a family of symmetric quadratic forms on a Hilbert space $\cH$. Assume that 
\begin{itemize}
\item[(a)] There exists a positive operator $B$ which commutes with $A$ and 
$$ CA \ge H(t) \ge C^{-1}A-CB \quad{and}\quad i[H(t),B] \le CA.$$

\item[(b)] The time-derivative of $H(t)$ exists and it is bounded in $Q(A)$: 
$$-C \langle x, A x\rangle \le \frac{d}{dt} \langle x, H(t) x \rangle \le C\langle x, A x\rangle\quad{for~all~}x\in Q(A).$$
\end{itemize}
Then for every $x_0\in Q(A)$, there exists a unique $x(t)\in L^\ii(0,1;Q(A))$ such that $t\mapsto x(t)$ is continuous in $\cH$ and 
$$
\left\{ \begin{gathered}
   i\dot{x}(t)= H(t)x(t)\quad {in}~~Q(A)^*,\hfill \\
   x(0)=x_0.\hfill \\ 
\end{gathered}  \right.
$$
Moreover, for all $t\in [0,1)$ we have $\|x(t)\|=\|x_0\|$ and 
$$\langle x(t), A x(t) \rangle \le Ce^{Ct} (\langle x_0, A x_0 \rangle +C).$$
\end{theorem}

In Theorem \ref{thm:dynamics-unbounded}, the condition (b) is similar to the assumption on the boundedness of the commutator $i[H(t),A]$, which is necessary to control the energy of the dynamics generated by a time-dependent Hamiltonian. A proof of Theorem \ref{thm:dynamics-unbounded} can be found in Appendix \ref{app:dynamics}. 

In order to apply Theorem \ref{thm:dynamics-unbounded}, we need the following bounds on the Bogoliubov Hamiltonian.

\begin{lemma}[Bogoliubov Hamiltonian] \label{le:Bog-Hamiltonian} Assume that $w^2 \le C(1-\Delta)$. Let $u_0\in H^1(\R^d)$ with $\|u_0\|=1$ and let $u(t)$ be the unique solution to the Hartree equation~\eqref{eq:Hartree-equation}. Then the Bogoliubov Hamiltonian $\bH(t)$ defined in (\ref{eq:Bogoliubov-Hamiltonian}) satisfies 
\begin{align} \label{eq:a-Bog-1}
C d\Gamma(1-\Delta) \ge \bH(t) \ge \dGamma (-\Delta-C).
\end{align}
on Fock space $\F$. Moreover, \[  
\pm \dot{\bH}(t) \le C d\Gamma(1-\Delta) \quad \text{and }  \pm i[\bH (t),\N ]\le C (\N+1). \]
\end{lemma} 

Assuming Lemma \ref{le:Bog-Hamiltonian} for the moment, we can give the

\begin{proof}[Proof of Theorem \ref{thm:Bogoliubov-dynamics}] Thanks to Lemma \ref{le:Bog-Hamiltonian}, we can apply Theorem \ref{thm:dynamics-unbounded} with $H(t)=\bH(t)$, $A=C\dGamma(1-\Delta)$ and $B=\cN+1$. Thus we obtain all desired results on $\Phi(t)$ except the fact that $\Phi(t)\in \F_+(t)$. To prove that $\Phi(t)\in \F_+(t)$, we compute  
\bqq
 \frac{d}{dt} \| a(u(t))\Phi(t)\|^2 &=& 2 \Re \left \langle  a(u(t))\Phi(t), \frac{d}{dt} \big(a(u(t)) \Phi(t) \big) \right\rangle \hfill\\
&=& 2 {\rm Im} \left \langle  a(u(t))\Phi(t), \big( -a(i\dot u(t)) + a(u(t)) \bH(t) \big) \Phi(t) \right\rangle.
\eqq
By using the Hartree equation $i\dot{u}_t=h(t) u(t)$ and the commutator relation
\bqq
[a(u(t)),\bH(t)] = [a(u(t)),\dGamma(h(t)) ]= a(h(t) u(t))  
\eqq
we get
\bqq
\frac{d}{dt} \| a(u(t))\Phi(t)\|^2 &=& 2{\rm Im} \left \langle  a(u(t))\Phi(t), \bH(t) a(u(t)) \Phi(t) \right\rangle =0.
\eqq
Since $a(u_0)\Phi_0=0$, we conclude that $a(u(t))\Phi(t)=0$, and hence $\Phi(t)\in \F_+(t)$ for all $t\ge 0$. 
\end{proof}

To prove Lemma \ref{le:Bog-Hamiltonian}, we need the following technical lemma. 

\begin{lemma}[Kernel estimates] \label{le:kernel} For $j=1,2$, the kernels $K_j(t,x,y)$ are bounded uniformly in time:
$$\|K_j(t,\cdot,\cdot)\|_{L^2(\R^d \times \R^d)}\le C, \quad \|\dot{K}_j(t,\cdot,\cdot)\|_{H^{-1}(\R^d \times \R^d)} \le C.$$
\end{lemma} 

Recall that $K_1(t)={Q(t)} \tilde K_1(t) Q(t)$ with $\tilde{K}_1(t,x,y)=u(t,x)w(x-y)\overline{u(t,y)}$ and $K_2(t)=(Q(t) \otimes Q(t)) \tilde K_2(t)$ with $\tilde K_2(t,x,y)=u(t,x)w(x-y)u(t,y)$. 

\begin{proof} From the assumption $w^2 \le C(1-\Delta)$ and from the uniform bound $\|u(t)\|_{H^1(\R^d)} \le C$, it follows that $\| w^2*|u(t)|^2\|_{L^\infty} \le C$. Thus 
\begin{align*} \int_{\R^d\times \R^d} |u(t,x)|^2 w(x-y)^2 |u(t,y)|^2 dx\,dy = \langle u(t), (w^2*|u(t)|^2)u(t) \rangle \le C 
\end{align*}
for $j=1,2$. The desired bounds on $K_j$'s follow immediately. 

Next, the time-derivatives $\dot{K}_j(t,x,y)$'s can be decomposed into many terms using the Hartree equation $i\dot{u}(t)=(-\Delta +|u(t)|^2*w-\mu(t))u(t)$. Let us consider the term $(\Delta_x u)(t,x)w(x-y) u(t,y)$ in detail.  Let $v\in H^1(\R^d \times \R^d)$. As in \cite[Lemma 6.2]{CheLeeSch-11}, by using an integration by part and the fact that $\nabla_x (w(x-y))=-\nabla_y (w(x-y))$, we get
\begin{align*} & \iint_{\R^d\times \R^d} (-\Delta_x u)(t,x)u(t,y) w(x-y) v(x,y) \, \d x \, \d y \nn\\
=&\iint_{\R^d\times \R^d} (\nabla_x u)(t,x)u(t,y) w(x-y) [\nabla_x v(x,y)+ \nabla_y v(x,y)]\, \d x \, \d y \nn\\
&+ \iint_{\R^d\times \R^d} (\nabla_x u)(t,x) (\nabla_y u)(t,y) w(x-y) v(x,y) \, \d x \, \d y.  \nn
\end{align*}
From the estimates $\| w^2*|u(t)|^2 \|_{L^\infty} \le C$, $\| w(x-y)v(x,y)\|_{L^2(\R^d \times \R^d)} \le C\|v\|_{H^1(\R^d \times \R^d)}$ and the Cauchy-Schwarz inequality we can conclude that
\begin{align*} \Big| \iint (-\Delta_x u)(t,x)u(t,y) w(x-y) v(x,y) \, \d x \, \d y \Big| \le  C \| v \|_{H^1}.
\end{align*}
Thus $(-\Delta_x u)(t,x)u(t,y) w(x-y)$ is bounded in $H^{-1}(\R^d \times \R^d)$. Since the other terms can be treated similarly, we obtain the desired bounds on $\dot{K}_j(t,x,y)$'s.
\end{proof}

Now we are able to give the

\begin{proof}[Proof of Lemma \ref{le:Bog-Hamiltonian}] First, we prove that,
\bq \label{eq:a-Bog-1b}
C d\Gamma(1-\Delta) \ge \bH(t) \ge \dGamma (1-\Delta)-C (\N +1)
.\eq
The term $d\Gamma(h(t))$ can be treated easily using the uniform bound $\| |u(t)|^2*w\|_{L^\infty} \le C$ and $\|K_1(t)\|_{L^2\to L^2}\le C$. Now let us consider the pairing term
$$\frac12\iint_{\R^d\times\R^d}\bigg(K_2(t,x,y) a^\dagger(x)a^\dagger(y) + \overline{K_2(t,x,y)} a(x)a(y) \bigg) \,\d x \, \d y .$$
For every state $\psi=(\psi_n)_{n\ge 0}$ in Fock space $\F=\bigoplus_{n\ge 0} \gH^n $ we can estimate  
\begin{align} \label{eq:paring-K2-start}
& \bigg| \Big\langle \psi , \iint_{\R^d\times\R^d} \d x \d y \, \overline{{K_2}(t,x,y)} a(x)a(y) \, \psi \Big\rangle \bigg| \nn\\
& = \bigg| \sum_{n\ge 0} \sqrt{(n+1)(n+2)} \iint_{\R^d\times\R^d}  \,\d x \,\d y \int_{\R^{dn}} \d{\bf z}  \overline{{K_2}(t,x,y)} \overline{\psi_{n}({\bf z})}\psi_{n+2}(x,y,{\bf z}) \bigg| \nn\\
&\le \sum_{n\ge 0} \sqrt{(n+1)(n+2)} \|K_2(t,.,.) \|_{L^2(\R^d \times \R^d)} \| \psi_n \| \| \psi_{n+2} \| \nn\\
& \le \frac{1}{2}\|K_2(t,.,.) \|_{L^2} \Big( \sum_{n\ge 0} (n+1)\| \psi_n \|^2 +  \sum_{n\ge 0} (n+2) \| \psi_{n+2} \|^2 \Big) \nn\\
&\le  \|K_2(t,.,.) \|_{L^2} \big\langle \psi, (\N+2) \psi \big \rangle.
\end{align}
Here we have introduced the variable ${\bf z}\in \R^{dn}$ and used the Cauchy-Schwarz inequality. From the kernel bound $\|K_2(t,.,.) \|_{L^2} \le C$ in Lemma \ref{le:kernel}, we can control the pairing term as
\begin{align} \label{eq:paring}
\pm \frac12\iint_{\R^d\times\R^d}\Big(K_2(t,x,y) a^\dagger(x)a^\dagger(y) +\text{h.c.}\Big) \, \d x \, \d y \leq C (\N+1).
\end{align}
Thus (\ref{eq:a-Bog-1}) follows immediately.

The time-derivative $\dot{\bH}(t)$ can be estimated in the same way, using the kernel bound $\|\dot{K}_j(t,.,.) \|_{H^{-1}} \le C$ instead of $\|K_j(t,.,.) \|_{L^2} \le C$. More precisely,  similarly to (\ref{eq:paring-K2-start})   we have now
\begin{align*}
& \bigg| \Big\langle \psi , \iint_{\R^d\times\R^d} \d x \, \d y \, \overline{{\dot{K_2}}(t,x,y)} a(x)a(y) \, \psi \Big\rangle \bigg| \nn\\
&\hspace{3cm}  \le \|\dot{K}_2(t,.,.) \|_{H^{-1}} \Big( \big\langle \psi, (\N+2) \psi \big \rangle + \big\langle \psi, d\Gamma(1-\Delta) \psi \big \rangle \Big)
\end{align*}
for every state $\psi\in Q(d\Gamma(1-\Delta))$. By similar bounds, we conclude that 
\bqq
\pm \dot{\bH}(t) \le C d\Gamma(1-\Delta).
\eqq

To prove the bound for the commutator $[\bH (t), \N]$ we observe that 
\begin{align} \label{eq:comm-pairing}i [\bH(t),\N ] &= \frac{i}{2}\iint_{\R^d\times\R^d}\bigg(K_2(t,x,y)[a^\dagger(x)a^\dagger(y),\N ] \nn\\
&\qquad\qquad\qquad\qquad+\overline{K_2(t,x,y)}[a(x)a(y),\N ]\bigg)\d x\,\d y \\
&= - i \iint_{\R^d\times\R^d} K_2(t,x,y) a^\dagger(x)a^\dagger(y) \d x \, \d y + \text{h.c.} \nn \end{align}
Similarly to (\ref{eq:paring-K2-start}), we find, for an arbitrary $\psi \in \F$,   
\[ \begin{split} \Big| \iint_{\R^d \times \R^d} K_2 (t,x,y) \,  \langle \psi,  \,a^\dagger(x) &a^\dagger(y)  \psi \rangle \, \d x \, \d y \Big| \\ & \leq\; C \| K_2 (t,.,.) \|_{L^2} \| (\N+1)^{1/2} \psi \|^2 .\end{split} \]
Since  $\| K_2 (t,.,.) \|_{L^2} \leq C$ by Lemma \ref{le:kernel}, we can conclude that $\pm [i\bH (t),\N] \leq C (\N+1)$.
\end{proof}

\section{Convergence of the fluctuations}\label{sec:reformulation}

\subsection{Statement involving $U_N(t)$} 

Now we reformulate the convergence in Theorem \ref{thm:fluctuations}. Let $w^2\le C(1-\Delta)$. Let $\Psi_{N}(t)=e^{itH_N}\Psi_{N,0}$ be the Schr\"odinger evolution with initial wave function $\Psi_{N,0}\in \bigotimes_{\rm sym}^N H^1(\R^d)$. Let $u(t)$ be the unique Hartree solution with initial datum $u_0\in H^1(\R^d)$. Let $\Phi(t)$ satisfy the Bogoliubov equation (\ref{eq:Bog-eq}) with initial datum $\Phi_0 \in Q(\dGamma(1-\Delta))\subset \F$ as in Theorem \ref{thm:Bogoliubov-dynamics} . Let $U_{N}(t)$ be defined from $u(t)$ as in (\ref{eq:def-unitary-UN}). Our main result in this section is 

\begin{theorem} [Convergence to Bogoliubov dynamics] \label{thm:fluctuations_U_N}  If 
\bq \label{eq:cv-initial-datum}
U_{N,0} \Psi_{N,0} \to \Phi_0
\eq
strongly in $\F$ and weakly in $Q(\dGamma(1-\Delta))$ as $N\to \infty$, then 
$$\boxed{\lim_{N\to \infty} U_{N}(t) \Psi_{N}(t)= \Phi(t)}$$
strongly in $\F$ and weakly in $Q(\dGamma(1-\Delta))$, for every fixed $t\geq0$.
\end{theorem}

\begin{remark} The condition (\ref{eq:cv-initial-datum}) is clearly satisfied if we start with a Hartree state at time 0, $\Psi_{N,0}=u_0^{\otimes N}$ with $u_0\in H^1(\R^d)$. In this case we have $U_{N,0}\Psi_{N,0}=|0 \rangle = \Phi_0$ for all $N$. 
\end{remark}

In Theorem~\ref{thm:fluctuations_U_N}, $\Phi(t)$ is a state on $\F_+(t)\subset \F$ which describes the fluctuations around the Hartree state $u(t)$ through the unitary $U_N(t)$. Note that
\begin{align*}
\lim_{N\to \infty} \| U_N(t) \Psi_N(t) - \Phi(t)\|_{\cF} &=\lim_{N\to \infty} \| U_N(t) \Psi_N(t) - \1_+^{\le N}\Phi_N(t)\|_{\cF}\\
&=\lim_{N\to \infty} \|  \Psi_N(t) - U_N(t)^* \1_+^{\le N} \Phi(t)\|_{\gH^N} \\
&=\lim_{N\to \infty}\norm{\Psi_N(t)-\sum_{n=0}^N u(t)^{\otimes (N-n)}\otimes_s\phi_n(t)}_{\gH^N}
\end{align*}
where $\1_{+}^{\le N}$ is the projection onto the truncated Fock space $\F_+(t)^{\le N}$. Therefore, Theorem \ref{thm:fluctuations_U_N} implies Theorem \ref{thm:fluctuations} immediately. Theorem \ref{thm:fluctuations_U_N} will be proved in the rest of this section. 

\subsection{Bound on number of particles outside of the condensate} The goal of this subsection is to estimate the expectation of the number of particles operator, and of its powers, in the evolved many body state $\phi_N (t) = U_N (t) \Psi_{N} (t)$. These estimates will play a crucial role in the proof of Theorem \ref{thm:fluctuations_U_N}. To control the growth of the expectation of the number of particles operator, we need the following bounds on the remainder $R_{N,1} (t)$ defined in (\ref{eq:RN1}).

\begin{lemma}[Remainder Hamiltonian] \label{le:remainder-bound} Let us denote $\1_{+}^{\le M}$ the projection onto the truncated Fock space $\F_+(t)^{\le M}$. For all $1\le M\le N$, we have
\begin{align} \label{eq:bound-ENt}
\pm \1_{+}^{\le M} R_{1,N} (t) \1_{+}^{\le M} \le C \sqrt{\frac{M}{N}}  (\N +1).
\end{align}
Moreover 
\begin{equation}\label{eq:bound-coR} \pm \1_{+}^{\le N} i\left[ R_{1,N} (t),\N \right] \1_{+}^{\le N} \le C \, (\N +1). \end{equation}
\end{lemma}

\begin{proof} The first estimate (\ref{eq:bound-ENt}) was already proved in \cite[Proposition 15]{LewNamSerSol-13}.

Now we consider the commutator $i[R_{1,N} (t),\N ]$. A straightforward computation similar to \eqref{eq:comm-pairing} shows that
\begin{align} \label{eq:comm-RN}
& i[R_{1,N}(t),\N] \nn \\
 &= - i a \big(Q(t) [w*|u(t)|^2]u(t)\big) \frac{\N_+(t) \, \sqrt{N-\N_+(t)}}{N-1} \nn\\
& +2i\bigg( \frac{\sqrt{(N-\N_+(t))(N-\N_+(t)-1)}}{N-1} -1 \bigg)\!  
\iint \overline{K_2(t,x,y)}a(x) a(y)\,\d x \,\d y \,  \nn\\
& + i\frac{\sqrt{N-\N_+(t)}}{N-1}  \iiiint  \big( \1 \otimes Q(t) \, w \, Q(t) \otimes Q(t) \big) (x,y; x',y')
\nn\\ &\hspace{5cm} \times
 u (t,x) a^\dagger(y) a(x')a(y')\,\d x\, \d y\, \d x'\, \d y' \nn\\
&+ \text{h.c.}
\end{align}
Recall that, in the last term, $w$ is understood as the multiplication operator by the function $w(x-y)$. 

We shall show that all terms in (\ref{eq:comm-RN}), when projected onto $\F_+(t)^{\le N}$, can be bounded by $C(\N+1)$. Let $v(t) = i Q(t) [w*|u(t)|^2] u(t)$. Since $\| v (t) \|_{L^2} \leq \| w * |u (t)|^2 \|_{L^\infty} \| u (t) \|_{L^2} \leq C$, the first term in (\ref{eq:comm-RN}) can be bounded, after projecting onto $\F_+(t)^{\le N}$, by 
\[ \pm \left\{ a (v(t)) \frac{\N_+(t) \, \sqrt{N-\N_+(t)}}{N-1}  + \text{h.c.}\right\}  \leq C (\N+1) \]
The second term on the right side of (\ref{eq:comm-RN}) can be treated similarly to the pairing term of $\bH(t)$ in (\ref{eq:comm-pairing}). As for the third term on the right side of (\ref{eq:comm-RN}), we write 
\begin{align}\label{eq:cub-1}
&\frac{\sqrt{N-\N_+}}{N-1} \iint\!\!\! \iiint \,\d x \,\d y \,\d y_1 \,\d x' \,\d y' Q_t (y,y_1) w (x-y_1) Q_t (x;x') Q_t (y_1;y') \nn\\ &\hspace{7cm} \times u (t,x) a^\dagger (y) a (x') a (y') \nn\\
&=\frac{\sqrt{N-\N_+}}{N-1} \iiint \,\d x \,\d y \,\d y'  w (x-y) Q_t (y;y') \, u (t,x) a^\dagger (Q_{t,y}) a (Q_{t,x}) a (y') \nn\\
&=\frac{\sqrt{N-\N_+}}{N-1} \iint \,\d x \,\d y\,  w (x-y)  \, u (t,x) a^\dagger (Q_{t,y}) a (Q_{t,x}) a (y) \\
&\hspace{.4cm} - \frac{\sqrt{N-\N_+}}{N-1} \iint \,\d x \,\d y\,    w (x-y) \, u (t,y) \, u (t,x) \, a^\dagger (Q_{t,y}) a (Q_{t,x}) a (u (t)) \nn
\end{align}
where we introduced the notation $Q_t (x,x')$ to denote the kernel of the operator $Q(t) = 1-|u (t) \rangle \langle u(t)|$, and where $Q_{t,x} (z) = Q_t (z,x)$. To bound the first term on the right side of (\ref{eq:cub-1}), we notice that, for arbitrary $\psi \in \F_+^{\le N}$, 
\[ \begin{split} 
\Big|  \iint & \d x \, \d y   \, w (x-y) \, u (t,x) \, \left\langle \psi,  \sqrt{N-\N_+} \, a^\dagger (Q_{t,y}) a (Q_{t,x}) a (y) \psi \right\rangle \Big| \\
&\leq \iint \d x\, \d y \, |w (x-y)| \, |u (t,x)| \left\| a(Q_{t,y}) \sqrt{N-\N_+} \psi \right\| \, \left\| a (y) a (Q_{t,x}) \psi \right\| .
\end{split} \]
By the Cauchy-Schwarz inequality we find
\begin{align*}
\Big|  \iint \d x\, &\d y   \, w (x-y) \, u (t,x) \, \left\langle \psi,  \sqrt{N-\N_+} \, a^\dagger (Q_{t,y}) a (Q_{t,x}) a (y) \psi \right\rangle \Big| \\
&\leq \iint \d x \d y \, |w (x-y)|^2 \, |u (t,x)|^2 \left\| a(Q_{t,y}) \sqrt{N-\N_+} \psi \right\|^2 \\
&\hspace{.4cm} + \iint \d x\, \d y \, \left\| a (y) a (Q_{t,x}) \psi \right\|^2 \\
&\leq C \int \d y  \left\| a(Q_{t,y}) \sqrt{N-\N_+} \psi \right\|^2 + \int dx \, \left\| a (Q_{t,x}) (\N+1)^{1/2} \psi \right\|^2 \\
&\leq CN \langle \psi, d\Gamma (Q (t)) \psi \rangle \leq C N \langle \psi , \N \psi \rangle .
\end{align*}
As for the second term on the right side of (\ref{eq:cub-1}), we have, for an arbitrary $\psi \in \F_+^{\le N}$, 
\[ \begin{split} 
\Big| \iint \d x\, \d y\, &w(x-y) \, u(t,x) u (t,y) \left\langle \psi, \sqrt{N-\N_+} \, a^\dagger (Q_{t,y}) a(Q_{t,x}) a(u_t) \psi \right\rangle \Big| \\
&\leq \iint \d x\, \d y \, |w (x-y)| |u (t,x)| |u (t,y)| \, \left\| a(Q_{t,y}) \sqrt{N-\N_+} \psi \right\| \\ &\hspace{7cm} \times  \left\| a(Q_{t,x}) a(u_t) \psi \right\| .
 \end{split} \]
By the Cauchy-Schwarz inequality again, we find
\[ \begin{split} 
\Big| \iint \d x \,\d y\, &w(x-y) \, u(t,x) u (t,y) \left\langle \psi, \sqrt{N-\N_+} \, a^\dagger (Q_{t,y}) a(Q_{t,x}) a(u_t) \psi \right\rangle \Big| \\
&\leq \iint \d x \d y \, |w (x-y)|^2 |u (t,x)|^2 \, \left\| a(Q_{t,y}) \sqrt{N-\N_+} \psi \right\|^2 \\ &\hspace{.4cm} +\iint \d x \d y \, |u (t,y)|^2 \,  \left\| a(Q_{t,x}) a(u_t) \psi \right\|^2 \\
&\leq C N \, \langle \psi, \cN \psi \rangle . \end{split} \]
Inserting in (\ref{eq:cub-1}), we conclude that the third term on the right side of (\ref{eq:comm-RN}), after projecting onto $\F_+^{\le N}$, is also bounded by $C (\cN+1)$.
\end{proof}

As an easy consequence of Lemma \ref{le:remainder-bound}, we have 
\begin{lemma}[Number of particles outside of the condensate] \label{le:N+-bound} There exists a constant $C$ such that
$$\langle \Phi_N(t), (\N+1)  \Phi_N(t) \rangle \le C e^{Ct} \langle \Phi_N(0), (\N+1) \Phi_N(0) \rangle.$$
\end{lemma}
\begin{proof} From the equation $i\dot{\Phi}_N(t)=(\bH(t)+R_N(t)) \Phi_N(t)$ and the commutator bounds in Lemma \ref{le:Bog-Hamiltonian} and Lemma \ref{le:remainder-bound}, we have 
\begin{align*} \frac{d}{dt} \langle \Phi_N(t), (\N +1) \Phi_N(t) \rangle &= \big\langle \Phi_N(t), i[\bH(t)+R_N(t),\N+1] \Phi_N(t) \big\rangle \\
&= \big\langle \Phi_N(t), i[R_{1,N}(t),\N] \Phi_N(t) \big\rangle \\
&\le C \langle \Phi_N(t), (\N +1) \, \Phi_N(t) \rangle.
\end{align*}
The desired bound then follows from Gronwall's inequality.
\end{proof}

\subsection{Proof of the main result}

Now we are able to prove Theorem \ref{thm:fluctuations_U_N} which itself implies Theorem~\ref{thm:fluctuations}. 

\begin{proof} We have
\begin{align}\label{eq:diff-1} 
\frac{d}{dt} &\| \Phi_{N}(t)-\Phi(t)\|^2 = 2\Re \Big\langle \Phi_{N}(t)-\Phi(t), \dot{\Phi}_{N}(t)-\dot \Phi(t)\Big\rangle \nn\\
&= -2 {\rm Im} \Big\langle \Phi_{N}(t)-\Phi(t), (\bH(t) +R_N(t)){\Phi}_{N}(t)-\bH(t) {\Phi}(t)\Big\rangle \nn\\
&= -2 {\rm Im} \Big\langle \Phi_{N}(t), R_N(t) {\Phi}(t)\Big\rangle \nn\\
&=-2  {\rm Im} \Big\langle \Phi_{N}(t), R_{1,N} (t) {\Phi}(t)\Big\rangle -2  {\rm Im} \Big\langle \Phi_{N}(t), R_{2,N} (t) {\Phi}(t)\Big\rangle  
\end{align}
where $R_{1,N}(t)$ and $R_{2,N}(t)$ are defined in (\ref{eq:RN1}) and (\ref{eq:RN2}), respectively. We shall show that both terms on the right side of (\ref{eq:diff-1}) tend to $0$ as $N\to \infty$. To this end, we denote by $\1_+^{\leq M}$ the orthogonal projection onto the truncated Fock space $\F_+(t)^{\leq M}$ and
$$\Phi(t)^{\le M}:=\1_+^{\le M} \Phi(t), \quad \Phi(t)^{> M}:=(\1-\1_+^{\le M}) \Phi(t) =\Phi(t) -\Phi(t)^{\le M}.$$ 
{\bf Estimate on the first remainder term.} We start by estimating the first term on the right side of (\ref{eq:diff-1}). Since $R_{1,N}$ can change the number of particles at most by two, for any $M \in \mathbb{N}$ we have
\[ \begin{split} 
\langle \Phi_N (t) , R_{1,N} (t) \Phi (t) \rangle = \, &\langle \Phi_N (t), R_{1,N} (t) \Phi (t)^{\leq M} \rangle + \langle \Phi_N (t), R_{1,N} (t) \Phi (t)^{> M} \rangle \\
= \; &\langle \Phi_N (t), \1_+^{\leq (M +2)} R_{1,N} (t) \1_+^{\leq (M+2)} \Phi (t)^{\leq M} \rangle \\ &+ \langle \Phi_N (t), \1^{\leq (N+2)} \, R_{1,N} (t)  \1^{\leq (N+2)} \Phi (t)^{> M} \rangle.
\end{split} \]
By Lemma \ref{le:remainder-bound}, we can choose $C>0$ so large that
\[ \begin{split} \1_+^{\leq (M +2)} R_{1,N} (t) \1_+^{\leq (M+2)} + C \sqrt{\frac{M}{N}} (\N+1) &\geq 0 \\ \1_+^{\leq (N+2)} R_{1,N} (t) \1_+^{\leq (N+2)} + C (\N+1) &\geq 0.
\end{split} \]
Then
\[ \begin{split} 
\langle \Phi_N &(t) , R_{1,N} (t) \Phi (t) \rangle \\ = \; &\left\langle \Phi_N (t), \left[ \1_+^{\leq (M +2)} R_{1,N} (t) \1_+^{\leq (M+2)}+ C \sqrt{\frac{M}{N}} \, (\N+1) \right] \Phi (t)^{\leq M} \right\rangle \\ &- C 
\sqrt{\frac{M}{N}} \, \left\langle \Phi_N (t), (\N+1) \Phi (t)^{\le M} \right\rangle \\
&+ \left\langle \Phi_N (t), \left[ \1^{\leq (N+2)} \, R_{1,N} (t)  \1^{\leq (N+2)} +C (\N+1) \right] \, \Phi (t)^{> M} \right\rangle \\ &- C \left\langle \Phi_N (t), (\N+1) \Phi (t)^{> M} \right\rangle.
\end{split} \]
By the Cauchy-Schwarz inequality, we estimate
\[ \begin{split}
\big| \langle &\Phi_N (t) , R_{1,N} (t) \Phi (t) \rangle \big| \\ \leq \; &\left\langle \Phi_N (t), \left[ \1_+^{\leq (M +2)} R_{1,N} (t) \1_+^{\leq (M+2)}+ C \sqrt{M/N} \, (\N+1) \right] \Phi_N (t) \right\rangle^{1/2} \\
&\hspace{.1cm} \times 
\left\langle \Phi^{\le M} (t), \left[ \1_+^{\leq (M +2)} R_{1,N} (t) \1_+^{\leq (M+2)}+ C \sqrt{M/N} \, (\N+1) \right] \Phi^{\le M} (t) \right\rangle^{1/2} \\
&+  \left\langle \Phi_N (t), \left[ \1^{\leq (N+2)} \, R_{1,N} (t)  \1^{\leq (N+2)} +C (\N+1) \right] \, \Phi_N (t) \right\rangle^{1/2} \\
&\hspace{.1cm} \times 
\left\langle \Phi (t)^{> M}, \left[ \1^{\leq (N+2)} \, R_{1,N} (t)  \1^{\leq (N+2)} +C (\N+1) \right] \, \Phi (t)^{> M} \right\rangle^{1/2} \\
&+C \sqrt{M/N} \, \left\langle \Phi_N (t), (\N+1) \Phi_N (t) \right\rangle^{1/2} \left\langle \Phi (t)^{\leq M}, (\N+1) \Phi (t)^{\leq M} \right\rangle^{1/2} \\
&+C \, \left\langle \Phi_N (t), (\N+1) \Phi_N (t) \right\rangle^{1/2} \left\langle \Phi (t)^{> M}, (\N+1) \Phi (t)^{> M} \right\rangle^{1/2} .
\end{split} \]
Using again Lemma \ref{le:remainder-bound} we find
\[ \begin{split} 
\big| \langle \Phi_N &(t) , R_{1,N} (t) \Phi (t) \rangle \big| \\ \leq \; &C \sqrt{M/N} \, \left\langle \Phi_N (t), (\N+1) \Phi_N (t) \right\rangle^{1/2}  \left\langle \Phi (t)^{\le M}, (\N+1) \Phi^{\le M} (t) \right\rangle^{1/2} \\ &+C \left\langle \Phi_N (t), (\N+1) \Phi_N (t) \right\rangle^{1/2}  \left\langle \Phi (t)^{> M}, (\N+1) \Phi^{> M} (t) \right\rangle^{1/2}. \end{split} \]
From Lemma \ref{le:N+-bound}, we conclude that
\[ \begin{split} 
\big| \langle \Phi_N &(t) , R_{1,N} (t) \Phi (t) \rangle \big| \\ \leq \; &C e^{Ct} \sqrt{M/N} \, \left\langle \Phi_N (0), (\N+1) \Phi_N (0) \right\rangle^{1/2}  \left\langle \Phi (0), (\N+1) \Phi (0) \right\rangle^{1/2} \\ &+C e^{Ct} \left\langle \Phi_N (0), (\N+1) \Phi_N (0) \right\rangle^{1/2}  \left\langle \Phi (t)^{> M}, (\N+1) \Phi^{> M} (t) \right\rangle^{1/2}. \end{split} \]
Letting first $N \to \infty$ and at the end $M \to \infty$, we obtain that 
\bq \label{eq:last-RN1}
\lim_{N\to \infty}\langle \Phi_N (t) , R_{1,N} (t) \Phi (t) \rangle= 0. 
\eq
{\bf Estimate on the second remainder term.} Next, we consider the second term on the right side of (\ref{eq:diff-1}). Using the shorthand notation $A = (Q(t) \otimes Q(t)) \, w \, (Q(t) \otimes Q(t))$, we have  
\[ \begin{split} \Big\langle \Phi_{N}(t), R_{2,N} (t) {\Phi}(t)\Big\rangle = \; &\frac{1}{N-1}\sum_{n\leq N}  \sum_{i<j}^n \left\langle \Phi^{(n)}_N (t), A_{ij} \, \Phi^{(n)} (t) \right\rangle \\ = \; &\frac{1}{N-1}\sum_{n \leq M} \frac{n(n-1)}{2} \, \left\langle \Phi^{(n)}_N (t), A_{12} \, \Phi^{(n)} (t) \right\rangle \\ &+ \frac{1}{N-1} \sum_{M \leq n \leq N} \frac{n(n-1)}{2} \, \left\langle \Phi^{(n)}_N (t), A_{12} \, \Phi^{(n)} (t) \right\rangle
\end{split}\]
where $A_{ij}$ denotes the two-particle operator $A$ acting on particles $i$ and $j$, and where we used the permutation symmetry of $\Phi_N (t)$ and $\Phi (t)$. By the Cauchy-Schwarz inequality, we find
\begin{align}\label{eq:quart-bd} 
\Big|  \Big\langle \Phi_{N}(t), R_{2,N} (t) {\Phi}(t)\Big\rangle \Big| \leq  &\frac{M^{3/2}}{N} \left( \sum_{n =0}^\infty n \| A_{12} \Phi^{(n)} (t) \|_2^2 \right)^{1/2} \\ & + \langle \Phi_N (t), \N \Phi_N (t) \rangle^{1/2}  \left( \sum_{n =M}^\infty n \| A_{12} \Phi^{(n)} (t) \|_2^2 \right)^{1/2}.\nn
\end{align}
Next, we observe that
\[ \begin{split} \| A_{12} \Phi^{(n)} (t) \|^2 &= \Big\langle \Phi^{(n)} (t), A_{12}^* A_{12} \Phi^{(n)} (t) \Big\rangle \\ &= \Big\langle \Phi^{(n)} (t),  Q_1 (t) Q_2 (t) \, w (x_1 - x_2) Q_1 (t) Q_2 (t) \\ &\hspace{4cm} \times w (x_1 - x_2) Q_1 (t) Q_2 (t) \Phi^{(n)} (t) \Big\rangle \\ 
&\leq \Big\langle  \Phi^{(n)} (t),  Q_1 (t) Q_2 (t) \, w^2 (x_1 - x_2) Q_1 (t) Q_2 (t) \Phi^{(n)} (t) \Big\rangle
\end{split} \]
where $Q_1 (t)$ and $Q_2 (t)$ denote the orthogonal projection $Q (t)$ acting on the first and the second particles, respectively. Using the assumption $w^2 \leq C(1-\Delta)$, we find
\[ \begin{split} 
 \| A_{12} \Phi^{(n)} (t) \|^2 
&\leq C \Big\langle  \Phi^{(n)} (t),  Q_1 (t) Q_2 (t) \, (1-\Delta_{x_1}) Q_1 (t) Q_2 (t) \Phi^{(n)} (t) \Big\rangle \\
&\leq C \Big\langle \Phi^{(n)} (t),  Q_1 (t)  \, (1-\Delta_{x_1}) Q_1 (t) \Phi^{(n)} (t) \Big\rangle.
\end{split} \]
We write $Q_1 (t) = \1 - P_1 (t)$, with $P (t) = |u(t) \rangle \langle u(t)|$ and $P_1 (t)$ denotes the orthogonal projection $P(t)$ acting on the first particle. We easily find
\[ \begin{split} 
\| A_{12} \Phi^{(n)} (t) \|^2  &\leq C \Big\langle \Phi^{(n)} (t), (1-\Delta_{x_1})  \Phi^{(n)} (t) \Big\rangle + \| u (t) \|_{H^1}^2 \| \Phi^{(n)} (t) \|_2^2 \\ &\leq C \Big\langle \Phi^{(n)} (t), (1-\Delta_{x_1})  \Phi^{(n)} (t) \Big\rangle.
\end{split} \]
Inserting the latter estimate into (\ref{eq:quart-bd}) and using Theorem \ref{thm:Bogoliubov-dynamics} and Lemma \ref{le:N+-bound}, we obtain
\begin{align*} 
&\Big|  \Big\langle \Phi_{N}(t), R_{2,N} (t) {\Phi}(t)\Big\rangle \Big| \nn\\
\leq  &\frac{M^{3/2}}{N} \Big\langle \Phi (t) , d\Gamma (1-\Delta) \Phi (t) \Big\rangle^{1/2} \nn\\ &+  \Big\langle \Phi_N (t), \N \Phi_N (t) \Big\rangle^{1/2} \Big\langle \Phi (t)^{>M} , d\Gamma (1-\Delta) \Phi (t)^{>M} \Big\rangle^{1/2}\nn \\
\leq \; & C e^{Ct} \frac{M^{3/2}}{N}   \Big\langle \Phi (0) , d\Gamma (1-\Delta) \Phi (0) \Big\rangle^{1/2} \nn\\ &+ C e^{Ct} \Big\langle \Phi_N (0), \N \Phi_N (0) \Big\rangle^{1/2} \Big\langle \Phi (t)^{>M} , d\Gamma (1-\Delta) \Phi (t)^{>M} \Big\rangle^{1/2}.
\end{align*}
Under the assumption (\ref{eq:cv-initial-datum}), we have 
$$\Big\langle \Phi (0) , d\Gamma (1-\Delta) \Phi (0) \Big\rangle \le C_0 \quad{\rm and}~~\Big\langle \Phi_N (0), \N \Phi_N (0) \Big\rangle \le C_0$$
where $C_0$ is dependent on the initial data, but independent of $N$. Therefore,
\begin{align*} 
&\Big|  \Big\langle \Phi_{N}(t), R_{2,N} (t) {\Phi}(t)\Big\rangle \Big| \nn\\
\leq \; & C e^{Ct} \frac{M^{3/2}}{N} +C e^{Ct}   \Big\langle \Phi (t)^{>M} , d\Gamma (1-\Delta) \Phi (t)^{>M} \Big\rangle^{1/2} .
\end{align*}
For every $t>0$ fixed, since $\langle \Phi (t), d\Gamma (1-\Delta) \Phi (t) \rangle \le C$ by Theorem \ref{thm:Bogoliubov-dynamics}, it follows that 
\bqq 
\lim_{M\to \infty} \langle \Phi^{>M} (t), d\Gamma (1-\Delta) \Phi^{>M} (t) \rangle= 0.
\eqq
Thus letting $N \to \infty$ and afterwards $M \to \infty$, we can conclude that
\bq \label{eq:last-RN2}
\lim_{N\to \infty} \Big\langle \Phi_{N}(t), R_{2,N} (t) {\Phi}(t)\Big\rangle =0 .
\eq
{\bf Conclusion.} Substituting (\ref{eq:last-RN1}) and (\ref{eq:last-RN2}) into (\ref{eq:diff-1}), we obtain
$$ \lim_{N\to \infty} \frac{d}{dt} \| \Phi_{N}(t)-\Phi(t)\|^2 =0.$$
Since $\lim_{N\to \infty} \| \Phi_{N}(0)-\Phi(0)\|^2=0$ by the assumption (\ref{eq:cv-initial-datum}), we conclude that 
$$\lim_{N\to \infty} \| \Phi_{N}(t)-\Phi(t)\|^2=0$$
for every $t\ge 0$ fixed. The proof is complete.
\end{proof}

\appendix

\section{Dynamics generated by quadratic forms} \label{app:dynamics}

In this appendix, we discuss the time-dependent Schr\"odinger equation
$$  i\dot{x}(t)= H(t)x(t)$$
where $H(t)$ is a family of symmetric quadratic forms in an abstract Hilbert space $\cH$. Recall that for a given self-adjoint operator $A \ge 1$ on $\cH$, we have the scale of spaces $Q(A)\subset \cH = \cH^* \subset Q(A)^*$, where $Q(A)$ is the quadratic form domain of $A$. Theorem \ref{thm:dynamics-unbounded} gives sufficient conditions on $H(t)$ and $A$ such that the Schr\"odinger equation $  \dot{x}(t)= -iH(t)x(t)$ can be solved in $Q(A)^*$. The proof below follows closely Simon's proof \cite[Theorem II.27]{Simon-71}, which goes back to Yoshida's argument.

\begin{proof}[Proof of Theorem \ref{thm:dynamics-unbounded}] We shall always denote by $\|.\|$ and $\langle ., . \rangle$ the norm and the inner product of $\cH$, and denote by $\dot{H}(t)$ the time derivative of $H(t)$, that is,
$$\frac{d}{dt} \langle x, H(t) x \rangle = \langle x,\dot{H}(t) x\rangle\quad~{\rm for~all~}x\in Q(A).$$

\medskip

\noindent {\bf Step 1: Construction of $x(t)$}. For every $n\in \mathbb{N}$, we introduce 
$$ H_n(t) :=P_n H(t)P_n~~{\rm with}~P_n:=(1+A/n)^{-1}.$$
Since $-CA \le  H(t) \le CA$, $H_n(t)$ can be extended to a bounded operator on $\cH$ with operator norm $\|H_n(t)\|\le Cn^{2}$. Moreover, since $-CA\le \dot{H}(t) \le CA$, the mapping $t\mapsto H_n(t)$ is Lipschitz continuous in the operator topology. Therefore, from every $x_0\in Q(A)\subset \cH$, we can find a unique evolution $x_n(t)\in \cH$ such that $x_n(0)=x_0$ and 
$$i\dot{x}_n(t)=H_n (t) x_n(t).$$

Now we estimate the energy $\langle x_n(t), A x_n(t) \rangle$. Let us introduce
$$ A_n(t):=P_n (H(t)+B)P_n=H_n(t)+ P_n B P_n \ge P_n A P_n.$$
Since $P_n$ commutes with $B$, we have
\begin{multline}
\frac{d}{dt} \langle x_n(t), A_n(t) x_n(t) \rangle = \langle x_n(t), \dot{A}_n(t) x_n(t) \rangle + \langle x_n(t), -i[A_n(t), H_n(t)] x_n(t) \rangle \hfill\\
=  \langle x_n(t), \dot{H}(t) x_n(t) \rangle + \langle x_n(t), -iP_n [B, H(t)] P_n x_n(t) \rangle .
\end{multline}
Moreover, since $$\dot{H}_n(t)\le CP_n A P_n \le CA_n(t)~~{\rm and}~-iP_n[B,H(t)]P_n\le CP_n A P_n\le CA_n(t)$$ we find that 
\bqq
\frac{d}{dt} \langle x_n(t), A_n(t) x_n(t) \rangle \le C \langle x_n(t), A_n(t) x_n(t) \rangle.
\eqq
From Gronwall's inequality, it follows that
\bqq
\langle x_n(t), A_n(t) x_n(t) \rangle \le  e^{Ct} \langle  x_0, A_n(0) x_0 \rangle \le  e^{Ct} \langle  x_0, P_n A P_n x_0 \rangle.
\eqq
Since $A_n(t) \ge P_n A P_n$ and $P_n \to 1$ strongly in $Q(A)$ as $n\to \infty$, we obtain
\bqq
 \limsup_{n\to \infty}\langle x_n(t), A x_n(t) \rangle \le e^{Ct} \langle x_0, A x_0 \rangle.
\eqq
Thus for $n$ large, $x_n(t)$ is bounded in $Q(A)$ uniformly in $n$ and $t$. Moreover,
\bqq
\langle \dot{x}_n(t), A^{-1} \dot{x}_n(t) \rangle &=& \big\langle  x_n(t) , H_n(t) A^{-1} H_n(t) x_n(t) \big\rangle \hfill\\
&=& \Big\langle A^{1/2} x_n(t) , (A^{-1/2}H_n(t) A^{-1/2})^2 A^{1/2} x_n(t) \Big\rangle
\eqq
Note that $\|A^{-1/2}H_n(t) A^{-1/2}\| \le C$ since $\|A^{-1/2}H(t) A^{-1/2}\| \le C$ and $P_n$ commutes with $A$. Therefore, we can conclude 
\bqq
\limsup_{n\to \infty }\langle \dot{x}_n(t), A^{-1} \dot{x}_n(t) \rangle &\le & \limsup_{n\to \infty} \|A^{1/2}x_n(t)\|^2 \| A^{-1/2}H_n(t) A^{-1/2}\|^2 \hfill\\
&\le & C e^{Ct}\|x_0\|_{Q(A)}^2 .
\eqq
Thus for $n$ large, $\dot{x}_n(t)$ is bounded in $Q(A)^*$ uniformly in $n$ and $t$. Consequently, up to a subsequence, we can assume that $x_n(t)$ converges to a limit $x(t)$ weakly in $L^\ii(0,1;Q(A))$ and that $\dot{x}_n(t)$ converges to $\dot{x}(t)$ weakly in $L^\ii(0,1;Q(A)^*)$. By~\cite{Lions-58}, we know that $x(t)\in C^0([0,1],\cH)$ and that $x(0)=\lim_{n\to\ii}x_n(0)=x_0$. Moreover, we have
$$\|x(t)\|_{Q(A)}^2 \le \liminf_{n\to \infty} \|x_n(t) \|_{Q(A)}^2 \le C e^{Ct}\|x_0\|_{Q(A)}^2.$$
On the other hand,  
\bqq
i\dot{x}_n(t)= H_n(t) x_n(t)= A^{1/2}P_n (A^{-1/2} H(t)A^{-1/2})P_n A^{1/2}x_n(t) \wto H(t) x(t)
\eqq
weakly in $Q(A)^*$ because $x_n(t) \wto x(t)$ weakly in $Q(A)$, $A^{-1/2}H(t)A^{-1/2}$ is bounded, and $P_n$ commutes with $A$ and converges to $1$ strongly in $Q(A)$. Therefore, by passing to the weak limit we obtain the desired equation
$$ i\dot{x}(t)= H(t) x(t)\quad{\rm in}~~Q(A)^*.$$

\medskip

\noindent{\bf Step 2: Uniqueness.} We need to show that if $x(t)\in L^\ii(0,1;Q(A))$ solves
$$
\left\{ \begin{gathered}
   i\dot{x}(t)= H(t)x(t)\quad {\rm in}~~Q(A)^*,\hfill \\
   x(0)=0\hfill \\ 
\end{gathered}  \right.
$$
then $x(t) \equiv 0$. Let $A(t):=H(t)+B \ge A$. Using the assumptions $\dot{H}(t)\le CA$ and $-i[H(t),B]\le CA$ we obtain
$$A(t)^{-1}\dot{A}(t) A(t)^{-1}=A(t)^{-1}\dot{H}(t) A(t)^{-1} \le CA(t)^{-1}A A(t)^{-1}\le CA(t)$$and 
$$-i[ A^{-1}(t),H(t)]=-A(t)^{-1}i[H(t),B]A(t)^{-1} \le CA(t)^{-1}A A(t)^{-1} \le CA(t).$$
Therefore, 
\begin{align*}
   \frac{d}{dt} \langle x(t), A^{-1}(t) x(t)\rangle &= \langle x(t), A(t)^{-1}\dot{A}(t) A(t)^{-1} x(t)\rangle \\
&~ \qquad + \langle x(t), -i[ A^{-1}(t),H(t)] x(t)\rangle \\
 &\le   C \langle x(t), A^{-1}(t) x(t)\rangle .
\end{align*}
If $x(0)=0$, then from Gronwall's inequality it follows that $x(t) \equiv 0$.
\end{proof}

\section{Derivation of \eqref{eq:UNHNUN-original}} \label{app:UN-HN-UN*}

For the reader's convenience, in this appendix we give a detailed calculation leading to \eqref{eq:UNHNUN-original}. 

Let us choose an orthonormal basis $\{v_n\}_{n=0}^\infty$ for $\gH$ with $v_n \in H^1(\R^d)$ for all $n\ge 0$ and $v_0=u(t)$. We shall denote $U_N=U_N(t)$, $\cN_+=\cN_+(t)$ and $a_n=a(v_n)$ for short. Similarly to \eqref{eq:2nd-dG1} and \eqref{eq:2nd-dG2}, we have 
\begin{align} \label{eq:2nd-dG1-af}
\bigoplus_{n\ge 1} \sum_{i=1}^n T_i = \dGamma(T)= \sum_{m,n\ge 0} T_{mn}a_m^\dagger a_n
\end{align}
and
\begin{align} \label{eq:2nd_dG2-af}
\bigoplus_{n\ge 2} \left( \sum_{1\le i<j \le n } w_{ij}   \right) = \frac{1}{2}\sum_{m,n,p,q\ge 0} W_{mnpq} a_m^\dagger (x)a_n^\dagger a_p a_q
\end{align}
where
\begin{align*}
T_{m,n} &= \langle u_m, -\Delta u_n \rangle = \int_{\R^d}  \overline{\nabla u_m (x)}\cdot \nabla u_n (x) dx, \\
W_{mnpq} &= \iint_{\R^d \times \R^d} u_m(x)u_n(y) w(x-y) u_p(x) u_q(y) dx dy.
\end{align*}
Now we compute $U_N H_N U_N^*$ using \eqref{eq:2nd-dG1-af} and \eqref{eq:2nd_dG2-af}:
\begin{align} \label{eq:UN-HN-UN-2nd-quantization-af}
U_N H_N U_N^* &= \sum_{m,n\ge 0} T_{mn} U_N a_m^\dagger a_n U_N^* \nn\\
&\quad +\frac{1}{2(N-1)} \sum_{m,n,p,q\ge 0} W_{mnpq} U_N a_m^\dagger (x)a_n^\dagger a_p a_q U_N^*.
\end{align}
For the kinetic part, we have immediately from \eqref{eq:UN-actions-1}-\eqref{eq:UN-actions-4} that
\begin{align*} 
U_N a_0^\dagger a_0 U_N^* &= N-\N_+,\\
U_N a_0^\dagger a_n U_N^* &= \sqrt{N-\cN_+} a_n ,\\
U_N a_m^\dagger a_0 U_N^* &= a_m^\dagger \sqrt{N-\cN_+},\\
U_N a_m^\dagger a_n U_N^* &= a_m^\dagger a_n.
\end{align*}
For the interaction part, we first use CCR \eqref{eq:CCR-af-ag} to get
\begin{align*}
U_N a_m^\dagger a_n^\dagger a_p a_q U_N^* &= U_N (a_m^\dagger a_p a_n^\dagger a_q - \delta_{np} a_m^\dagger a_q U_N^* \\
&= \Big( U_N a_m^\dagger a_p U_N^* \Big) \Big( U_N a_n^\dagger a_q U_N^* \Big) - \delta_{np}  U_N a_m^\dagger a_q U_N^* 
\end{align*}
for all $m,n,p,q\ge 0$ and then apply \eqref{eq:UN-actions-1}-\eqref{eq:UN-actions-4}. It results
\begin{align*}
U_N a_0^\dagger a_0^\dagger a_0 a_0 U_N^* &= (N-\cN_+)^2 - (N-\cN_+) \\
&= (N-\cN_+)(N-\cN_+-1) \\
U_N a_0^\dagger a_0^\dagger a_0 a_q U_N^* & =   (N-\cN_+)^{3/2} a_q - \sqrt{N-\cN_+} a_q \\
&= (N-\cN_+-1) \sqrt{N-\cN_+} a_q \\
U_N a_0^\dagger a_0^\dagger a_p a_q U_N^* &= \sqrt{N-\cN_+} a_p \sqrt{N-\cN_+} a_q \\
&= \sqrt{N-\cN_+}\sqrt{N-\cN_+ - 1} a_p a_q \\
U_N a_0^\dagger a_n^\dagger a_0 a_q U_N^* &= (N-\cN_+)a_n^\dagger a_q \\
U_N a_0^\dagger a_n^\dagger a_p a_q U_N^* & = \sqrt{N-\cN_+} a_p a_n^\dagger a_q - \delta_{np} \sqrt{N-\cN_+} a_q \\
&= \sqrt{N-\cN_+} a_n^\dagger a_p  a_q \\
U_N a_m^\dagger a_n^\dagger a_p a_q U_N^* &= a_m^\dagger a_p a_n^\dagger  a_q - \delta_{np} a_m^\dagger a_q = a_m^\dagger a_n^\dagger a_p a_q
\end{align*}
for all $m,n,p,q \ge 1$. All in all, inserting the above identities (and their adjoints) into  \eqref{eq:UN-HN-UN-2nd-quantization-af} we find that (c.f. eq. (4.8) in \cite{LewNamSerSol-13})

\begin{align} \label{eq:UN-HN-UN-a}
U_N H_N U_N^* &= T_{00} (N-\cN_+) + \frac12 W_{0000}\frac{(N-\cN_+)(N-\cN_+-1)}{N-1} \nn\\
& + \Big[ \sum_{m\ge 1}\left( T_{0m} + W_{000m} \frac{N-\N_+ -1} {N-1} \right) \sqrt{N-\N_+} a_m +{\rm h.c.} \Big]\nn\\
&+\sum_{m,n\ge 1}  \Big( W_{0m0n}+ W_{0mn0} \Big) \frac{N-\N_+} {N-1} a_m^\dagger a_n  \nn\\
&+ \Big[ \frac{1}{2} \sum_{m,n\ge 1} W_{mn00}  a_m^\dagger a_n^\dagger \frac{\sqrt{(N-\N_+)(N-\N_+-1)}}{N-1} + {\rm h.c.} \Big],\nn\\
&+ \Big[ \frac{1}{N-1} \sum_{m,n,p \ge 1} W_{0npq} \sqrt{N-\N_+} a_n^\dagger a_p a_q+{\rm h.c.} \Big], \nn\\
& +  \sum_{m,n\ge 1} T_{mn} a_m^* a_n + \frac{1}{2(N-1)} \sum_{m,n,p,q \ge 1} W_{mnpq} a_m ^* a_n ^* a_p a_q.
\end{align}
Here we have used $W_{000m}=W_{00m0}$, $W_{0m0n}=W_{m0n0}$ and $W_{0mn0}=W_{m00n}$. Also, we have written $\Big[ X+{\rm h.c.} \Big]$ instead of $\Big[ X+X^* \Big]$ for short. 

Now we show that the right side of \eqref{eq:UN-HN-UN-a} coincides with the right side of \eqref{eq:UNHNUN-original}, where the latter has the advantage that there is no need to use any basis $\{v_n\}_{n=0}^\infty$ in its representation. First, 
\begin{align} \label{eq:sum-1}
&\quad T_{00} (N-\cN_+) + \frac12 W_{0000}\frac{(N-\cN_+)(N-\cN_+-1)}{N-1} \nn\\
&= \Big( T_{00} + \frac{1}{2}W_{0000} \Big) N - \Big( T_{00} + W_{0000} \Big) \cN_+ + W_{0000}\frac{\cN_+(\cN_+-1)}{2(N-1)}\nn\\
&= e(t) N - \Big( e(t) + \mu(t) \Big) \cN_+ + \mu(t) \frac{\cN_+(\cN_+-1)}{N-1} \nn\\
&= e(t) N - \Big( e(t) + \mu(t) \Big) \dGamma(Q(t)) + \mu(t) \dGamma(Q(t)) \frac{\cN_+-1}{N-1}
\end{align} 
Here we have used $e(t)=T_{00}+\frac{1}{2}W_{0000}$, $\mu(t)= \frac{1}{2}W_{0000}$ and $\cN_+=\dGamma(Q(t))$. Next, similarly to \eqref{eq:2nd-dG1-af}, we get
\begin{align*}
\sum_{m,n\ge 1} T_{mn} a_m^\dagger a_n &= \sum_{m,n\ge 1}  \langle u_m, T u_n \rangle a_m^* a_n  \nn\\
&= \sum_{m,n\ge 0}  \langle u_m, Q(t)T Q(t) u_n \rangle a_m^* a_n   = \dGamma \Big( Q(t) T Q(t)\Big).
\end{align*}
and
\begin{align*}
\sum_{m,n\ge 1}  \Big( W_{0m0n}+ W_{0mn0} \Big) a_m^\dagger a_n &=  \sum_{m,n\ge 1} \langle u_m, [|u(t)|^2*w + K_1(t)] u_n \rangle  a_m^\dagger a_n \\
& =  \dGamma(Q(t) [|u(t)|^2*w + K_1(t)]Q(t)).
\end{align*}
Therefore, 
\begin{align} \label{eq:sum-2}
&\quad\sum_{m,n\ge 1} T_{mn} a_m^\dagger a_n + \sum_{m,n\ge 1}  \Big( W_{0m0n}+ W_{0mn0} \Big) \frac{N-\cN_+}{N-1}a_m^\dagger a_n \nn\\
&= \dGamma \Big( Q(t) T Q(t)\Big)+  \dGamma(Q(t) [|u(t)|^2*w + K_1(t)]Q(t))\nn\\
&= \dGamma \Big( Q(t) [h(t)+K_1(t)+\mu(t)]Q(t)\Big) \frac{N-\cN_+}{N-1} \nn\\
&\quad +   \dGamma(Q(t) [|u(t)|^2*w + K_1(t)]Q(t)) \frac{1-\cN_+}{N-1}.
\end{align}
On the other hand, from
\begin{align*}
\sum_{m\ge 1} T_{0m} a_m &= \sum_{m\ge 1} \langle u(t), T v_m \rangle a(v_m) \\
&= a\Big( \sum_{m\ge 1} \langle v_m, T u(t) \rangle v_m\Big)= a\Big( Q(t) (-\Delta) u(t)\Big) 
\end{align*}
and 
\begin{align*}
\sum_{m\ge 1} W_{000m} a_m = \sum_{m\ge 1} \langle u(t), (|u(t)|^2*w) u_m \rangle a_m =  a\Big( Q(t) [|u(t)|^2*w] u(t)\Big),
\end{align*}
it follows that
\begin{align} \label{eq:sum-3}
&\sum_{m\ge 1}\left( T_{0m} + W_{000m} \frac{N-\N_+ -1} {N-1} \right) \sqrt{N-\N_+} a_m \nn\\
&= \sqrt{N-\N_+} \sum_{m\ge 1}\left( T_{0m} + W_{000m} \right) a_m - \frac{\cN_+\sqrt{N-\N_+}}{N-1} \sum_{m\ge 1} W_{000m} a_m \nn \\
&= \sqrt{N-\N_+} a\Big( Q(t)h(t)u(t)\Big)-\frac{\cN_+\sqrt{N-\N_+}}{N-1} a\Big( Q(t) [|u(t)|^2*w] u(t)\Big).
\end{align}
Here in the first term of the right side of \eqref{eq:sum-3}, we have used $$Q(t)[-\Delta+|u(t)|^2*w]u(t)=Q(t)h(t)u(t),$$
which follows from the definition $h(t)=-\Delta+|u(t)|^2*w-\mu(t)$ and $Q(t)u(t)=0$. Finally, thanks to the relations \eqref{eq:af-ax}, we have the equivalent representations
\begin{align} \label{eq:sum-4}
&\quad \frac{1}{2} \sum_{m,n\ge 1} W_{mn00}  a_m^\dagger a_n^\dagger \frac{\sqrt{(N-\N_+)(N-\N_+-1)}}{N-1} \nn\\
&= \frac{1}{2} \iint dx\, dy\, K_2(t,x,y)a^\dagger(x) a^\dagger(y)   \frac{\sqrt{(N-\N_+)(N-\N_+-1)}}{N-1} ,\\
\label{eq:sum-5}
&\frac{\sqrt{N-\N_+}}{N-1} \sum_{m,n,p \ge 1} W_{0npq}  a_n^\dagger a_p a_q \nn\\
&=\frac{\sqrt{N-\N_+(t)}}{N-1} \iiiint  \, \big( \1 \otimes Q(t) \, w \, 
Q(t) \otimes Q(t) \big) (x,y; x', y') \times \nn\\ & \quad\quad\quad\quad\quad\quad\quad\quad\quad\quad\quad\quad\quad \times u (t,x) \, a^\dagger(y) a(x')a(y')\,\d x \, \d y \, \d x'\, \d y'
\end{align}
and
\begin{align} \label{eq:sum-6}
&\frac{1}{2(N-1)} \sum_{m,n,p,q \ge 1} W_{mnpq} a_m ^* a_n ^* a_p a_q \nn\\
&= \frac{1}{2(N-1)} \iiiint  \big( Q(t)\otimes Q(t) \, w \, Q(t) \otimes Q(t) \big) (x,y; x',y') \nn \\ & \quad\quad\quad\quad\quad\quad\quad\quad\quad\quad\quad \times a^\dagger(x) a^\dagger(y) a(x')a(y')\,\d x\, \d y\, \d x'\, \d y' .
\end{align}
Inserting \eqref{eq:sum-1}-\eqref{eq:sum-6} into the right side of \eqref{eq:UN-HN-UN-a}, we obtain   
\begin{align*}
U_N H_N U_N^* & = e(t) N - \dGamma \Big( Q(t) [h(t)+K_1(t)-e(t)]Q(t)\Big) \\ 
& + \Big[ \sqrt{N-\N_+} a\Big( Q(t)h(t)u(t)\Big) + {\rm h.c.} \Big] \\
& + \dGamma(Q(t) [|u(t)|^2*w + K_1(t)-\mu(t)]Q(t)) \frac{1-\cN_+}{N-1}  \\
& - \Big[ \frac{\cN_+\sqrt{N-\N_+}}{N-1} a\Big( Q(t) [|u(t)|^2*w] u(t)\Big) + {\rm h.c.} \Big]\\\
& + \Big[ \frac{1}{2} \iint  K_2(t,x,y)a^\dagger(x) a^\dagger(y)\,\d x\, \d y\, \\
& \quad \quad \quad \quad \quad \quad \times    \frac{\sqrt{(N-\N_+)(N-\N_+-1)}}{N-1} + {\rm h.c.} \Big] \\
& + \Big[ \frac{\sqrt{N-\N_+(t)}}{N-1} \iiiint  \, \big( \1 \otimes Q(t) \, w \, 
Q(t) \otimes Q(t) \big) (x,y; x', y') \nn\\ & \quad\quad\quad\quad\quad\quad\quad \times u (t,x) \, a^\dagger(y) a(x')a(y')\,\d x\, \d y\, \d x'\, \d y' + {\rm h.c.} \Big] \\
& + \frac{1}{2(N-1)} \iiiint  \big( Q(t)\otimes Q(t) \, w \, Q(t) \otimes Q(t) \big) (x,y; x',y') \nn \\ & \quad\quad\quad \quad\quad \quad\quad \quad \times a^\dagger(x) a^\dagger(y) a(x')a(y')\,\d x\, \d y\, \d x'\, \d y' 
 \end{align*}
which is exactly \eqref{eq:UNHNUN-original}. 



\end{document}